\documentclass[journal]{IEEEtran}

\usepackage{cite}
\usepackage[pdftex]{graphicx}
\makeatletter
\let\MYcaption\@makecaption
\makeatother

\usepackage[font=footnotesize]{subcaption}

\makeatletter
\let\@makecaption\MYcaption
\makeatother
\usepackage{amsmath,amssymb,amsfonts,bm}
\usepackage{amsthm}
\usepackage{chngcntr}
\newtheorem{thm}{Theorem}
\newtheorem{cor}{Corollary}

\newtheorem{asmp}{Assumption}
\newtheorem{cond}{Condition}
\newtheorem{lem}{Lemma}
\theoremstyle{definition}
\newtheorem{rem}{Remark}

\usepackage[ruled,norelsize]{algorithm2e}
\SetKwBlock{Repeat}{repeat}{}
\usepackage{array}
\usepackage{bbm}
\usepackage{comment}
\usepackage{hhline}
\usepackage[flushleft]{threeparttable}
\usepackage{nccmath}
\usepackage{multicol}
\usepackage{makecell}
\usepackage{tabularx}
\usepackage[hidelinks]{hyperref}
\usepackage{xcolor}

\newcolumntype{Y}{>{\centering\arraybackslash}X}
\newcolumntype{b}{>{\hsize=1.35\hsize}Y}
\newcolumntype{z}{>{\hsize=1.2\hsize}Y}
\newcolumntype{s}{>{\hsize=.45\hsize}Y}

\makeatletter
\newcommand*{\inlineequation}[2][]{%
  \begingroup
    \refstepcounter{equation}%
    \ifx\\#1\\%
    \else
      \label{#1}%
    \fi
    \relpenalty=10000 %
    \binoppenalty=10000 %
    \ensuremath{%
      #2%
    }%
    ~\@eqnnum
  \endgroup
}
\makeatother
\makeatletter
\newcommand{\removelatexerror}{\let\@latex@error\@gobble}
\makeatother

\newcounter{parentcond}
\makeatletter 
\newenvironment{subcond}[1]{%
  \counterwithin*{cond}{parentcond}
  \def\subcondcounter{#1}%
  \refstepcounter{#1}%
  \protected@edef\theparentcond{\csname the#1\endcsname}%
  \setcounter{parentcond}{\value{#1}}%
  \setcounter{#1}{0}%
  \expandafter\def\csname the#1\endcsname{\theparentcond\alph{#1}}%
  \ignorespaces
}{%
  \setcounter{\subcondcounter}{\value{parentcond}}%
  \counterwithout*{cond}{parentcond} 
  \ignorespacesafterend
}
\makeatother

\makeatletter
\let\save@mathaccent\mathaccent
\newcommand*\if@single[3]{%
  \setbox0\hbox{${\mathaccent"0362{#1}}^H$}%
  \setbox2\hbox{${\mathaccent"0362{\kern0pt#1}}^H$}%
  \ifdim\ht0=\ht2 #3\else #2\fi
  }
\newcommand*\rel@kern[1]{\kern#1\dimexpr\macc@kerna}
\newcommand*\widebar[1]{\@ifnextchar^{{\wide@bar{#1}{0}}}{\wide@bar{#1}{1}}}
\newcommand*\wide@bar[2]{\if@single{#1}{\wide@bar@{#1}{#2}{1}}{\wide@bar@{#1}{#2}{2}}}
\newcommand*\wide@bar@[3]{%
  \begingroup
  \def\mathaccent##1##2{%
    \let\mathaccent\save@mathaccent
    \if#32 \let\macc@nucleus\first@char \fi
    \setbox\z@\hbox{$\macc@style{\macc@nucleus}_{}$}%
    \setbox\tw@\hbox{$\macc@style{\macc@nucleus}{}_{}$}%
    \dimen@\wd\tw@
    \advance\dimen@-\wd\z@
    \divide\dimen@ 3
    \@tempdima\wd\tw@
    \advance\@tempdima-\scriptspace
    \divide\@tempdima 10
    \advance\dimen@-\@tempdima
    \ifdim\dimen@>\z@ \dimen@0pt\fi
    \rel@kern{0.6}\kern-\dimen@
    \if#31
      \overline{\rel@kern{-0.6}\kern\dimen@\macc@nucleus\rel@kern{0.4}\kern\dimen@}%
      \advance\dimen@0.4\dimexpr\macc@kerna
      \let\final@kern#2%
      \ifdim\dimen@<\z@ \let\final@kern1\fi
      \if\final@kern1 \kern-\dimen@\fi
    \else
      \overline{\rel@kern{-0.6}\kern\dimen@#1}%
    \fi
  }%
  \macc@depth\@ne
  \let\math@bgroup\@empty \let\math@egroup\macc@set@skewchar
  \mathsurround\z@ \frozen@everymath{\mathgroup\macc@group\relax}%
  \macc@set@skewchar\relax
  \let\mathaccentV\macc@nested@a
  \if#31
    \macc@nested@a\relax111{#1}%
  \else
    \def\gobble@till@marker##1\endmarker{}%
    \futurelet\first@char\gobble@till@marker#1\endmarker
    \ifcat\noexpand\first@char A\else
      \def\first@char{}%
    \fi
    \macc@nested@a\relax111{\first@char}%
  \fi
  \endgroup
}
\makeatother

\SetKwFor{At}{at}{do}{end}

\begin{document}

\title{Distributed Adaptive Signal Fusion for Fractional Programs}

\author{Cem Ates~Musluoglu,
        and~Alexander~Bertrand,~\IEEEmembership{Senior~Member,~IEEE}
\thanks{Copyright \copyright 2025 IEEE. Personal use of this material is permitted. Permission from IEEE must be obtained for all other uses, in any current or future media, including reprinting/republishing this material for advertising or promotional purposes, creating new collective works, for resale or redistribution to servers or lists, or reuse of any copyrighted component of this work in other works.}
\thanks{This project has received funding from the European Research Council (ERC) under the European Union's Horizon 2020 research and innovation programme (grant agreement No 802895 and grant agreement 101138304). The authors also acknowledge the financial support of the FWO (Research Foundation Flanders) for project G081722N, and the Flemish Government (AI Research Program). Views and opinions expressed are however those of the author(s) only and do not necessarily reflect those of the European Union or the granting authorities. Neither the European Union nor the granting authorities can be held responsible for them. }
\thanks{C.A. Musluoglu and A. Bertrand are with KU Leuven, Department of Electrical Engineering (ESAT), Stadius Center for Dynamical Systems, Signal Processing and Data Analytics, Kasteelpark Arenberg 10, box 2446, 3001 Leuven, Belgium and with Leuven.AI - KU Leuven institute for AI. e-mail: cemates.musluoglu, alexander.bertrand @esat.kuleuven.be}
\thanks{A conference precursor of this manuscript has been published in \cite{musluoglu2022fdasf}.}
}

\maketitle

\begin{abstract}
The distributed adaptive signal fusion (DASF) is an algorithmic framework that allows solving spatial filtering optimization problems in a distributed and adaptive fashion over a bandwidth-constrained wireless sensor network. The DASF algorithm requires each node to sequentially build a compressed version of the original network-wide problem and solve it locally. However, these local problems can still result in a high computational load at the nodes, especially when the required solver is iterative. In this paper, we study the particular case of fractional programs, i.e., problems for which the objective function is a fraction of two continuous functions, which indeed require such iterative solvers. By exploiting the structure of a commonly used method for solving fractional programs and interleaving it with the iterations of the standard DASF algorithm, we obtain a distributed algorithm with a significantly reduced computational cost compared to the straightforward application of DASF as a meta-algorithm. We prove convergence and optimality of this ``fractional DASF'' (F-DASF) algorithm and demonstrate its performance via numerical simulations.
\end{abstract}


\section{Introduction}

\IEEEPARstart{W}{ireless} sensor networks (WSNs) consist of a set of sensor nodes that are distributed over an area to collect sensor signals at different locations. In a signal processing context, the aim of a WSN is often to find a filter to apply to these signals gathered at the various nodes such that the resulting filtered signals are optimal in some sense, or satisfy certain desired properties. In an adaptive context where signal statistics are varying in time, centralizing the data is very costly in terms of bandwidth and energy resources, as it requires continuously communicating a large number of samples of signals measured at the nodes. Therefore, various methods have been proposed, such as consensus, incremental strategies, diffusion or the alternating direction method of multipliers (ADMM) \cite{olfati2005consensus,lopes2007incremental,chen2012diffusion,boyd2011distributed} to find these optimal filters in a fully distributed fashion, i.e., where the different sensor nodes collaborate with each other to perform a certain inference task.

In this paper, we are interested in signal fusion and spatial filtering problems, commonly encountered in biomedical signal analysis \cite{bertrand2015distributed,blankertz2007optimizing,ramoser2000optimal}, wireless communication \cite{bjornson2020scalable, sanguinetti2019toward,nayebi2016performance}, or acoustics \cite{furnon2021distributed,zhang2018rate,benesty2008microphone}. More specifically, our goal is to solve such problems over WSNs, by finding an optimal, network-wide spatial filter $X$, to optimally fuse the channels of the network-wide multi-channel signal $\mathbf{y}(t)$ obtained from stacking the local signals measured at each node of the network. Mathematically, the optimal $X$ results from solving an optimization problem of which the objective function is written in the form $\varrho(X^T\mathbf{y}(t))$.

We focus in particular on so-called fractional programs, i.e., problems where the function $\varrho$ can be written as a fraction of two continuous functions, i.e., $\varrho(X^T\mathbf{y}(t))=\varphi_1(X^T\mathbf{y}(t))/\varphi_2(X^T\mathbf{y}(t))$. Several spatial filtering problems have such an objective, including signal-to-noise ratio (SNR) maximization filters \cite{van1988beamforming}, trace ratio optimization (TRO) \cite{wang2007trace}, regularized total least squares (RTLS) \cite{sima2004regularized,beck2006finding}, among others, and have found applications in the aforementioned fields. For example, in the context of electroencephalography (EEG) sensor networks, the motor imagery problem aims to find such a filter $X$ by solving the TRO problem to distinguish between imagined left versus right-hand movement using signals collected from various sensors placed on the scalp \cite{blankertz2007optimizing,wang2007trace}. Another example consists of solving the RTLS problem in an antenna or microphone array for direction-of-arrival estimation when both the observations and the source signals are noisy \cite{zhu2010sparse}.

In the existing distributed methods mentioned earlier, the objective $\varrho$ is assumed to be per-node separable, i.e., can be written as a sum $\varrho(X)=\sum_k \varrho_k(X)$ over the nodes $k$ of the network, where $\varrho_k$ depends solely on the data of node $k$. However, the objective function of the problems of interest, namely fractional spatial filtering problems, is not per-node separable and requires second-order statistics between all the sensor channel pairs of the nodes in the network. Artificially rewriting these problems as consensus-type ones to be able to apply ADMM or consensus is still possible but leads to a large increase in communication costs over the network that does not scale with the network size, even leading to nodes sharing more data than what they collect \cite{musluoglu2022unifiedp1,hovine2024nonsmooth}. This is because these algorithms iterate over a fixed batch of samples (each node transmitting the same batch multiple times), and have to start from scratch for each new batch of incoming samples. Therefore, we are interested in distributed methods that are directly designed to solve spatial filtering problems in an adaptive setting with streaming data to avoid having such issues.

For this purpose, an algorithmic framework called distributed adaptive signal fusion (DASF) has been proposed in \cite{musluoglu2022unifiedp1,musluoglu2022unifiedp2}. The main idea behind this method is to iteratively create a compressed version of the global problem at each node by transmitting only compressed data within the network, therefore greatly reducing the communication cost. These compressed problems are then solved locally at each node, and the process is repeated with a new batch of measured signals to be able to track changes in the signal statistics. A convenient property of the DASF ``meta'' algorithm is that the local problems to be solved at each node are compressed instances of the original (centralized) problem. The DASF algorithm therefore merely requires to ``copy-paste'' this solver at each individual node \cite{musluoglu2022dsfotoolbox}. 

However, solving the compressed problem can be computationally expensive when an expensive iterative solver is required. This is indeed also the case for fractional programs, which are commonly solved using an iterative algorithm referred to as Dinkelbach's procedure. In this paper, we propose the fractional DASF (F-DASF) algorithm which implements a single iteration of the Dinkelbach procedure within each iteration of the DASF algorithm, therefore interleaving the iterations of both algorithms and avoiding nested iterative loops. This however implies that F-DASF cannot rely on the convergence guarantees and proofs of the original DASF algorithm in \cite{musluoglu2022unifiedp1,musluoglu2022unifiedp2}, as these assume that the nodes fully execute Dinkelbach's procedure until convergence within each iteration of DASF. In contrast, the F-DASF algorithm only partially solves the compressed problem at each iteration by applying only a single step of the Dinkelbach procedure. As we will see, this will imply that various subresults in the convergence proof applicable to the DASF algorithm cannot be applied to the F-DASF algorithm. Results such as the fact that fixed points of the algorithm are stationary, and even the monotonic decrease of the objective have to be proven from scratch. Therefore, the main result of this paper will be to show that the simplified F-DASF algorithm still converges to the solution of the global problem, under similar conditions (such as well-posedness and constraint qualifications) as the original DASF algorithm. In our results, we will see that on top of significantly reducing the computational cost and simplifying the overall method, the F-DASF algorithm has comparable convergence rates to the DASF algorithm. 

The resulting F-DASF algorithm can in fact be viewed as a generalization of the distributed trace ratio (TRO) algorithm in \cite{musluoglu2021distributed,musluoglu2020dtro}, which is a particular fractional program where both $\varphi_1$ and $\varphi_2$ consist of only a single quadratic term. Our proposed F-DASF algorithm can be applied to the entire class of fractional programs. Furthermore, it extends the ``vanilla'' DASF framework in \cite{musluoglu2022unifiedp1,musluoglu2022unifiedp2} towards a more efficient class of algorithms for the particular case of fractional programs.

The outline of this paper is as follows. In Section \ref{sec:frac_prog}, we review the basics of fractional programs and Dinkelbach's procedure to solve them. In Section \ref{sec:prob_setting} we formally define the distributed problem setting and the assumptions used throughout this paper. The proposed F-DASF algorithm is provided in Section \ref{sec:fdasf} while its convergence proof is given in Section \ref{sec:convergence}. Finally, Section \ref{sec:simulations} shows the performance of the algorithm and its comparison with the state-of-the-art.

\textbf{Notation:} Uppercase letters are used to represent matrices and sets, the latter in calligraphic script, while scalars, scalar-valued functions, and vectors are represented by lowercase letters, the latter in bold. We use the notation $\chi_q^i$ to refer to a certain mathematical object $\chi$ (such as a matrix, set, etc.) at node $q$ and iteration $i$. The notation $\left(\chi^i\right)_{i\in\mathcal{I}}$ refers to a sequence of elements $\chi^i$ over every index $i$ in the ordered index set $\mathcal{I}$. If it is clear from the context, we omit the index set $\mathcal{I}$ and simply write $\left(\chi^i\right)_i$. A similar notation $\{\chi^i\}_{i\in\mathcal{I}}$ is used for non-ordered sets. We define an accumulation point of a sequence $(X^i)_{i\in\mathbb{N}}$ as the limit of a converging subsequence $(X^i)_{i\in\mathcal{I}}$ of $(X^i)_{i\in\mathbb{N}}$, with $\mathcal{I}\subseteq \mathbb{N}$. Additionally, $I_Q$ denotes the $Q\times Q$ identity matrix, $\mathbb{E}[\cdot]$ the expectation operator, $\text{tr}(\cdot)$ the trace operator, and $|\cdot|$ the cardinality of a set.

\section{Fractional Programming Review}\label{sec:frac_prog}
A fractional program is an optimization problem with an objective function $r$ represented by a ratio of two continuous and real-valued functions $f_1$ and $f_2$:
\begin{equation}\label{eq:prob_frac}
  \begin{aligned}
    \underset{X}{\text{minimize } } \quad & r(X)\triangleq\frac{f_1(X)}{f_2(X)}\\
    \textrm{subject to} \quad & X\in\mathcal{S},\\
    \end{aligned}
\end{equation}
where $\mathcal{S}\subset\mathbb{R}^{M\times Q}$ is a non-empty constraint set and $f_2(X)>0$ for $X\in\mathcal{S}$. We define the minimal value of $r$ over $\mathcal{S}$ as $\rho^*\triangleq \min_{X\in\mathcal{S}}r(X)$ and the set of arguments achieving this value as $\mathcal{X}^*\triangleq \{X\in\mathcal{S}\;|\;r(X)=\rho^*\}$. We denote by $X^*\in\mathcal{X}^*$ one particular solution of (\ref{eq:prob_frac}). To solve fractional programs, there exist two prominent methods based on solving auxiliary problems instead of the original problem (\ref{eq:prob_frac}), for which an overview can be found in \cite{schaible1983fractional}. The first method consists of creating an equivalent problem by defining new optimization variables using what is referred to as the Charnes-Cooper transform. It initially was presented in \cite{charnes1962programming} for linear $f_1$ and $f_2$ with $\mathcal{S}$ a convex polyhedron, and was then extended to cases where convexity assumptions were imposed on $f_1$ and $f_2$ \cite{schaible1974parameter,schaible1983fractional}. The second method is a parametric approach and is the one we will focus on in this paper. Initially presented in \cite{dinkelbach1967nonlinear}, it is often referred to as Dinkelbach's procedure. Let us define the auxiliary functions $f:\mathcal{S}\times \mathbb{R}\rightarrow\mathbb{R}$ and $g:\mathbb{R}\rightarrow\mathbb{R}$:
\begin{align}
  f(X,\rho)&\triangleq f_1(X)-\rho f_2(X),\label{eq:auxiliary}\\
  g(\rho)&\triangleq \min_{X\in\mathcal{S}} f(X,\rho), \label{eq:min_auxiliary}
\end{align}
where $\rho$ is a real-valued scalar. It can be shown that $g$ is continuous, strictly decreasing with $\rho$, and has a unique root. We also observe that for any fixed $X$, the function $f$ in (\ref{eq:auxiliary}) is affine, and therefore $g$ is concave. Note that these results do not require convexity assumptions on $f_1$ or $f_2$. It has been shown \cite{jagannathan1966some,dinkelbach1967nonlinear,schaible1976fractional,crouzeix1985algorithm,crouzeix1991algorithms} that there is a direct relationship between the roots of $g$ and the solution of (\ref{eq:prob_frac}), given in the following lemma.
\begin{lem}[see \cite{crouzeix1985algorithm}]\label{lem:frac_min}
  Suppose $\mathcal{S}$ is compact. Then, if for a scalar $\rho$ we have $g(\rho)=0$ then $\rho=\rho^*$, where $\rho^*$ is the global minimum of $r$.
\end{lem}
\noindent It is noted that it is possible to obtain weaker relationships between the roots of $g$ and (\ref{eq:prob_frac}) if $\mathcal{S}$ is not compact \cite{crouzeix1991algorithms,crouzeix1985algorithm}, but this is beyond the scope of this review, i.e., we will assume that $\mathcal{S}$ is compact in the remaining of this paper, unless mentioned otherwise.

Dinkelbach's procedure is an iterative method aiming to find the unique root $\rho^*$ of $g$. We start by initializing $X^0$ randomly and set $\rho^0=r(X^0)$. Then, at each iteration $i$, we find the solution set of the auxiliary problem minimizing $f$ given $\rho$:
\begin{equation}\label{eq:aux_prob}
  \begin{aligned}
    \underset{X}{\text{minimize } } \quad & f(X,\rho^i)=f_1(X)-\rho^if_2(X)\\
    \textrm{subject to} \quad & X\in\mathcal{S}.\\
    \end{aligned}
\end{equation}
The solution of (\ref{eq:aux_prob}) might not be unique but given by a set $\mathcal{X}^{i+1}$, in which case one $X^{i+1}\in\mathcal{X}^{i+1}$ is selected. Finally, the function values are updated as:
\begin{equation}
  \rho^{i+1}=r(X^{i+1}).
\end{equation}
The steps of Dinkelbach's procedure are summarized in Algorithm \ref{alg:dinkelbach}, while convergence results are summarized in the following theorem, which combines results from \cite{crouzeix1985algorithm,crouzeix1991algorithms,crouzeix1986note}.
\begin{thm}\label{thm:dinkelbach_convergence}
  Assuming $\mathcal{S}$ is compact, the sequence $(\rho^i)_i$ converges to the finite minimal value $\rho^*$ of (\ref{eq:prob_frac}). Additionally, convergent subsequences of $(X^i)_i$ converge to an optimal solution of (\ref{eq:prob_frac}). If Problem (\ref{eq:prob_frac}) has a unique solution $X^*$, then $(X^i)_i$ converges to $X^*$.
\end{thm}
If the constraint set $\mathcal{S}$ is not compact, we can still obtain convergence of the sequence $(\rho^i)_i$ to $\rho^*$ if Problem (\ref{eq:prob_frac}) has a solution, the solution sets of the auxiliary problems (\ref{eq:aux_prob}) are not empty and $\text{sup}_i f_2(X^i)$ is finite \cite{crouzeix1985algorithm}, with $\text{sup}$ denoting the supremum.

\begin{figure}[!t]
  \removelatexerror
  \begin{algorithm}[H]
  \caption{Dinkelbach's procedure \cite{dinkelbach1967nonlinear}}\label{alg:dinkelbach}
  \DontPrintSemicolon
  \SetKwInOut{Input}{input}\SetKwInOut{Output}{output}
  \Output{$\rho^*$}
  \BlankLine
  $X^0$ initialized randomly in $\mathcal{S}$, $\rho^0\gets r(X^0),\; i\gets0$\;
  \Repeat
  {
  1) $X^{i+1}\gets \underset{X\in\mathcal{S}}{\text{argmin }}f(X,\rho^i)$\;
  2) $\rho^{i+1}\gets r(X^{i+1})$ where $r$ is defined in (\ref{eq:prob_frac})\;
  
  $i\gets i+1$\;
  }
  \end{algorithm}
\end{figure}

It is noted that Dinkelbach's procedure can be viewed as an instance of Newton's root-finding method, as we can show that \cite{crouzeix1991algorithms}:
\begin{equation}
  \rho^{i+1}=\rho^{i}-\frac{g(\rho^{i})}{-f_2(X^{i+1})},
\end{equation}
which corresponds to Newton's root finding method applied to the function $g$, since $-f_2(X^{i+1})$ is a subgradient of $g$ at $\rho^i$. Other algorithms using a different root finding method have also been described to solve fractional programs, for example the bisection method \cite{schaible1983fractional,beck2006finding}. It is however indicated in \cite{crouzeix2008revisiting} that the Newton root finding procedure has faster convergence compared to those in the context of fractional programming.

\section{Problem Setting and Preliminaries}\label{sec:prob_setting}

We consider a WSN represented by a graph $\mathcal{G}$ consisting of $K$ nodes where the set of nodes is denoted as $\mathcal{K}=\{1,\dots,K\}$. Each node $k$ measures a stochastic $M_k-$channel signal $\mathbf{y}_k$ considered to be (short-term) stationary and ergodic. The observation of $\mathbf{y}_k$ at time $t$ is denoted as $\mathbf{y}_k(t)\in\mathbb{R}^{M_k}$. In practice, we will often omit this time index $t$ for notational convenience unless we explicitly want to emphasize the time-dependence of $\mathbf{y}$. We define the network-wide signal to be
\begin{equation}\label{eq:partition_y}
  \mathbf{y}=[\mathbf{y}_1^T,\dots,\mathbf{y}_K^T]^T,
\end{equation}
such that the time sample $\mathbf{y}(t)$ at time $t$ is the $M-$dimensional vector obtained by stacking every $\mathbf{y}_k(t)$, where $M=\sum_{k\in\mathcal{K}}M_k$. In our algorithm development, we do not assume the statistics of $\mathbf{y}$ to be known, which implies that the algorithm has to estimate or learn these on the fly based on incoming sensor data at the different nodes.

\subsection{Fractional Problems for Signal Fusion in WSNs}

In spatial filtering and signal fusion, we aim to combine the channels of the network-wide signal $\mathbf{y}$ using a linear filter $X$ such that the filter $X$  optimizes a particular objective function. In this work, we are interested in spatial filters that are the solution of some fractional program. Two well-known examples, namely trace ratio optimization (TRO) and regularized total least squares (RTLS), are shown in Table \ref{tab:ex_prob} for illustration purposes.

A generic instance of such fractional programs for spatial filtering and signal fusion can be written as
\begin{equation}\label{eq:prob_g}
  \begin{aligned}
  \underset{X\in\mathbb{R}^{M\times Q}}{\text{minimize } } \quad & \varrho(X^T\mathbf{y}(t),X^TB)\triangleq\frac{\varphi_1\left(X^T\mathbf{y}(t),X^TB\right)}{\varphi_2\left(X^T\mathbf{y}(t),X^TB\right)}\\
  \textrm{subject to} \quad & \eta_j\left(X^T\mathbf{y}(t),X^TB\right)\leq 0,\;\textrm{ $\forall j\in\mathcal{J}_I$,}\\
    & \eta_j\left(X^T\mathbf{y}(t),X^TB\right)=0,\;\textrm{ $\forall j\in\mathcal{J}_E$.}
  \end{aligned}
\end{equation}
The functions $\eta_j$ represent inequality constraints for $j\in\mathcal{J}_I$ and equality constraints for $j\in\mathcal{J}_E$, while the full index set of constraints is denoted as $\mathcal{J}=\mathcal{J}_I\cup\mathcal{J}_E$, such that we have $J=|\mathcal{J}|$ constraints in total. Throughout this paper, we assume that the constraint set $\mathcal{S}$ of (\ref{eq:prob_g}) is compact, as is commonly assumed for fractional programs to guarantee convergence of the Dinkelbach procedure. We also assume that the denominator $\varphi_2$ is always strictly positive over $\mathcal{S}$ (or strictly negative, in which case we can minimize $-\varphi_1/-\varphi_2$). The matrix $B$ is an $M\times L$ deterministic matrix independent of time. We will see that $\mathbf{y}$ and $B$ are treated similarly in the proposed method, however, we make the distinction between the two to emphasize the adaptive properties of the proposed algorithm. In practice, the matrices $B$ are often required to enforce some constraints and their number of columns $L$ is application dependent. For example, in the TRO problem of Table \ref{tab:ex_prob}, we have $B=I_M$ in the constraints. In a distributed context, we partition $B$ similar to (\ref{eq:partition_y}):
\begin{equation}\label{eq:partition_B}
  B=[B_1^T,\dots,B_K^T]^T,
\end{equation}
where $B_k$ is assumed to be known by node $k$. Note that the optimization variable $X$ in (\ref{eq:prob_g}) only appears in the fusion form $X^T\mathbf{y}$ or $X^TB$, which is a requirement for the DASF framework. Furthermore, every other parameter of the problem that is not fused with $X$ is assumed to be known by every node (such as the signal $d$ in the RTLS example in Table \ref{tab:ex_prob}). 

The optimization problems written in the form (\ref{eq:prob_g}) represent a subclass of the problems considered in the DASF framework \cite{musluoglu2022unifiedp1}, namely problems for which the objective is a fractional function. Our goal is to derive a new algorithm for these cases, which is computationally more attractive than straightforwardly applying the generic DASF ``meta'' algorithm to (\ref{eq:prob_g}).

Similar to the DASF framework in \cite{musluoglu2022unifiedp1}, we also allow Problem (\ref{eq:prob_g}) to contain multiple signals $\mathbf{y}$ or matrices $B$, which have been omitted from (\ref{eq:prob_g}) for notational conciseness. For example, the RTLS problem in Table \ref{tab:ex_prob} requires two different $B-$matrices: $B^{(1)}=I_M$ in the denominator of the objective function since $||\mathbf{x}||^2 =(\mathbf{x}^T\cdot I_M)\cdot(\mathbf{x}^T\cdot I_M)^T=(\mathbf{x}^T\cdot B^{(1)})\cdot (\mathbf{x}^T\cdot B^{(1)})^T$, and $B^{(2)}=A$ in the constraint function. For the sake of intelligibility, we will continue the derivation and algorithm analysis for a single set of $\mathbf{y}$, and $B$, yet we emphasize that this is not a hard restriction as the algorithm can be easily generalized to multiple signal variables and multiple deterministic matrices (we refer to the original DASF paper \cite{musluoglu2022unifiedp1} for more details on these cases).

Note that in a centralized context, Problem (\ref{eq:prob_g}) can be solved using the Dinkelbach procedure, where each auxiliary problem can be written as
\begin{equation}\label{eq:aux_prob_g}
  \begin{aligned}
    \underset{X\in\mathbb{R}^{M\times Q}}{\text{minimize } } \quad & \varphi_1\left(X^T\mathbf{y}(t),X^TB\right)-\rho^i\;\varphi_2\left(X^T\mathbf{y}(t),X^TB\right)\\
    \textrm{subject to} \quad & \eta_j\left(X^T\mathbf{y}(t),X^TB\right)\leq 0,\;\textrm{ $\forall j\in\mathcal{J}_I$,}\\
      & \eta_j\left(X^T\mathbf{y}(t),X^TB\right)=0,\;\textrm{ $\forall j\in\mathcal{J}_E$,}
    \end{aligned}
\end{equation}
with
\begin{equation}
  \rho^i=\frac{\varphi_1\left(X^{iT}\mathbf{y}(t),X^{iT}B\right)}{\varphi_2\left(X^{iT}\mathbf{y}(t),X^{iT}B\right)}.
\end{equation}
Since the minimization of Problem (\ref{eq:prob_g}) is done over $X$ only, we define the following functions in order to compactify some of the equations throughout this paper:
\begin{align}\label{eq:r_f_h}
  r(X)&\triangleq \varrho(X^T\mathbf{y}(t),X^TB),\nonumber\\
  f_j(X)&\triangleq \varphi_j(X^T\mathbf{y}(t),X^TB),\; j\in\{1,2\},\\
  h_j(X)&\triangleq \eta_j(X^T\mathbf{y}(t),X^TB),\; \forall j\in\mathcal{J} \nonumber.
\end{align}
These definitions allow us to write Problem (\ref{eq:prob_g}) as in (\ref{eq:prob_frac}) and the auxiliary problems (\ref{eq:aux_prob_g}) as in (\ref{eq:aux_prob}), where $\mathcal{S}$ denotes the constraint set defined by the constraint functions $h_j$. In the remaining parts of this text, we will mention interchangeably (\ref{eq:prob_frac}) and (\ref{eq:prob_g}) to refer to the problem we aim to solve while interchangeably using (\ref{eq:aux_prob}) and (\ref{eq:aux_prob_g}) for its auxiliary problems, depending on whether we want to emphasize the specific $X^T\mathbf{y}(t)$, or $X^TB$, structure or not. We define $X^*$ to be a solution of Problem (\ref{eq:prob_g}) and $\mathcal{X}^*$ to be the full solution set.

\begin{table}
  \renewcommand{\arraystretch}{2}
  \caption{Examples of problems with fractional objectives as in (\ref{eq:prob_g}).}
  \label{tab:ex_prob}
  \begin{tablenotes}
    \item \textit{In TRO, $\mathbf{v}$ is a second stochastic signal collected by the nodes in addition to $\mathbf{y}$. In RTLS, $d$ is a target signal that is assumed to be known, and $A$ is a deterministic matrix used for Tikhonov regularization.}
  \end{tablenotes}
  \begin{tabularx}{\columnwidth}{|s|b|z|}
  \hline
   Problem & Cost function to minimize & Constraints \\ \hhline{|=|=|=|}
   TRO \cite{wang2007trace} & $-\frac{\mathbb{E}[||X^T\mathbf{v}(t)||^2]}{\mathbb{E}[||X^T\mathbf{y}(t)||^2]}$ & $X^TX=I_Q$ {\scriptsize $\Leftrightarrow (X^TB)(X^TB)^T=I_Q$ with $B=I_M$} \\ \hline
   RTLS \cite{sima2004regularized,beck2006finding} & $\frac{\mathbb{E}[||\mathbf{x}^T\mathbf{y}(t)-d(t)||^2]}{1+||\mathbf{x}||^2}$ & $||\mathbf{x}^TA||^2\leq \delta^2$ \\ \hline
  \end{tabularx}
\end{table}

\subsection{General Assumptions}\label{sec:assumptions}

To be able to state convergence guarantees of the proposed algorithm, we require some assumptions on Problem (\ref{eq:prob_g}). These assumptions are similar to those for the DASF algorithm proposed in \cite{musluoglu2022unifiedp1,musluoglu2022unifiedp2}, the main difference being that some assumptions are required to hold for the auxiliary problems as well, in order for our proposed method to work. In practice, if Problem (\ref{eq:prob_g}) satisfies these assumptions, the corresponding auxiliary problems typically do so as well.

\begin{asmp}\label{asmp:well_posed}
  The targeted instance of Problem (\ref{eq:prob_g}) and its corresponding auxiliary problems (\ref{eq:aux_prob_g}) are well-posed\footnote{The notion of (generalized Hadamard) well-posedness we require is based on \cite{hadamard1902problemes,zhou2005hadamard}. The main difference is that we require the map from the parameter (inputs of the problem) space to the solution space to be continuous instead of upper semicontinuous, which is required for the convergence proof.}, in the sense that the solution set is not empty and varies continuously with a change in the parameters of the problem.
\end{asmp}
\noindent This first assumption translates to requiring that an infinitesimal change in the input of the problem, for example in the signal statistics of $\mathbf{y}$, results in an infinitesimal change in its solutions $X^*$. As an example, consider the ``full-rank'' versus ``rank-deficient'' least squares problems, for which Assumption \ref{asmp:well_posed} is satisfied for the former but not the latter.

\begin{asmp}\label{asmp:kkt}
  The solutions of Problem (\ref{eq:prob_g}) and its corresponding auxiliary problems (\ref{eq:aux_prob_g}) satisfy the linear independence constraint qualifications (LICQ), i.e., $X^*$ satisfies the Karush-Kuhn-Tucker (KKT) conditions.
\end{asmp}
\noindent Assumption \ref{asmp:kkt} is a common assumption in the field of optimization \cite{peterson1973review}. The LICQ enforces that, if $X^*$ is a solution of Problem (\ref{eq:prob_g}), then the gradients $\nabla_X h_j(X^*)$, $j\in\mathcal{J}^*$, are linearly independent\footnote{A set of matrices $\{A_j\}_{j\in\mathcal{J}}$ is linearly independent when $\sum_{j\in\mathcal{J}}\alpha_jA_j=0$ is satisfied if and only if $\alpha_j=0$, $\forall j\in\mathcal{J}$, or equivalently, when $\{\text{vec}(A_j)\}_{j\in\mathcal{J}}$ is a set of linearly independent vectors, where $\text{vec}(\cdot)$ is the vectorization operator.}, where $\mathcal{J}^*\subseteq \mathcal{J}$ is the set of all indices $j$ for which $h_j(X^*)=0$. This is a requirement for the KKT conditions to be necessary for a solution $X^*$. Further details on constraint qualifications can be found in \cite{peterson1973review}. If there is no constraint function $\eta_j$ in Problem (\ref{eq:prob_g}), Assumption \ref{asmp:kkt} simply implies that $\nabla_X f(X^*)=0$.

An additional assumption requiring the compactness of the sublevel sets of the objective of (\ref{eq:aux_prob_g}) is required in the original DASF algorithm \cite{musluoglu2022unifiedp1,musluoglu2022unifiedp2}, its purpose being to ensure that the points generated by the algorithm lie in a compact set. As we will show in Section \ref{sec:convergence}, this assumption can be omitted in the case of the proposed F-DASF algorithm due to the fact that we only consider fractional programs with compact constraint sets. Note that the latter assumption should not be viewed as a limiting condition, as we will empirically demonstrate in Section \ref{sec:qol} through a simulation that the F-DASF algorithm introduced in Section \ref{sec:fdasf} can still converge to the correct solution for non-compact constraint sets in problems for which the centralized Dinkelbach procedure also converges, yet without a theoretical guarantee.

\subsection{Adaptivity and Approximations of Statistics}\label{sec:adaptivity}

As mentioned previously, the signal $\mathbf{y}$ is stochastic, therefore the functions $\varphi_1,\varphi_2,\eta_j$ implicitly contain an operator to transform probabilistic quantities into deterministic values, such that the problem solved in practice is deterministic. The most common operator encountered in signal processing applications is the expectation, i.e., there exist deterministic functions $\Phi_1,\Phi_2,H_j$ such that
\begin{align}
  \varphi_j(X^T\mathbf{y}(t),X^TB)&=\mathbb{E}[\Phi_j(X^T\mathbf{y}(t),X^TB)],\;j\in\{1,2\},\nonumber\\
  \eta_j(X^T\mathbf{y}(t),X^TB)&=\mathbb{E}[H_j(X^T\mathbf{y}(t),X^TB)],\;\forall j\in\mathcal{J}.
\end{align}
A common way to approximate these functions in a practical implementation is to take $N$ time samples of the signal $\mathbf{y}$, say $\{\mathbf{y}(\tau)\}_{\tau=t}^{t+N-1}$, and approximate the expectation through a sample average, for example:
\begin{equation}\label{eq:approx_E}
  \mathbb{E}[\Phi_j(X^T\mathbf{y}(t),X^TB)]\approx \frac{1}{N}\sum_{\tau=t}^{t+N-1}\Phi_j(X^T\mathbf{y}(\tau),X^TB).
\end{equation}
Under the ergodicity assumption on the signal $\mathbf{y}$, (\ref{eq:approx_E}) gives an accurate approximation for sufficiently large $N$. In practical realizations, the objective functions and constraints in (\ref{eq:prob_g}) will be evaluated on such sample batches of size $N$. 

Furthermore, the stationarity assumption imposed on $\mathbf{y}$ is merely added for mathematical tractability, as it allows us to remove the time-dependence, as it is often done in the adaptive signal processing literature to analyze asymptotic convergence \cite{diniz2019adaptive,sayed2011adaptive}. However, in practice, we do not require the signals to be stationary, but rather short-term stationary and/or with slowly varying statistics. In Section \ref{sec:simulations}, we will demonstrate that the proposed method is indeed able to track changes in the statistical properties of the signal. For ease of notation, we will use the exact functions $\varphi_j$ and $\eta_j$ (or $f_j$ and $h_j$) in the remaining parts of this text, while in practice and in the simulations presented in Section \ref{sec:simulations} they will be replaced by an approximation such as (\ref{eq:approx_E}).

\section{Cost-Efficient DASF for Fractional Programs}\label{sec:fdasf}

In this section, we propose a new algorithm for solving (\ref{eq:prob_g}) in a distributed and adaptive fashion, where the required statistics of $\mathbf{y}$ are estimated and learned on the fly based on the incoming sensor observations. As mentioned earlier, the DASF algorithm in \cite{musluoglu2022unifiedp1} can be straightforwardly applied to (\ref{eq:prob_g}) to obtain such a distributed algorithm, where the nodes would use a solver for the original (centralized) problem to solve locally compressed instances of the network-wide problem. However, this would require a node to fully execute Dinkelbach's procedure (Algorithm \ref{alg:dinkelbach}) within each iteration of the DASF algorithm. Since Algorithm \ref{alg:dinkelbach} is itself iterative, we would obtain nested iterations of Algorithm \ref{alg:dinkelbach} within the DASF algorithm's iterations, which would result in a large computational cost at each node. Instead, we propose the fractional DASF (F-DASF) algorithm interleaving the steps of the DASF algorithm with the iterations of Dinkelbach's procedure (Algorithm \ref{alg:dinkelbach}), where each DASF iteration only requires to perform \textit{a single} iteration of Algorithm \ref{alg:dinkelbach}, thereby greatly reducing the computational burden.

\subsection{Dynamic Network Pruning}\label{sec:pruning}

At each iteration $i$, an updating node $q\in\mathcal{K}$ is selected to be the so-called updating node. As will be described in the next subsections, the other nodes $k\neq q$ of the network will forward and fuse compressed data towards the updating node $q$ which will be responsible for performing the update on the variable $X$ at the current iteration. The selection of the updating node for different iterations will be done in a round-robin fashion. Based on the selected node, the network is pruned so as to obtain a spanning tree $\mathcal{T}^i(\mathcal{G},q)$, such that there is a unique path connecting each pair of nodes. The pruning's main purpose is to define a spanning tree with a root in the updating node $q$, allowing to aggregate the data at this node. For networks with arbitrary topologies, feedback loops leading to unwanted effects in the data fusion process would be present if no pruning were implemented. The pruning function $\mathcal{T}^i$ can be chosen freely, however, it should avoid cutting any link between the updating node $q$ and its neighbors $n\in\mathcal{N}_q$ \cite{musluoglu2022unifiedp2}, where we denote as $\mathcal{N}_q$ the set of nodes neighboring node $q$. A visual example of a pruning $\mathcal{T}^i$ for a given graph $\mathcal{G}$ is provided in Figure \ref{fig:pruning}. In practice, the pruning can be implemented in a distributed fashion within the network, e.g., as in Algorithm $4$ of \cite{hovine2022maxvar}. We further note that the dependence of $\mathcal{T}^i$ on the index $i$ is for flexibility purposes, e.g., if the network connectivity graph changes from iteration to iteration. In the remaining parts of this section, the set $\mathcal{N}_k$ refers to the neighbors of node $k$ in the pruned network, i.e., with respect to the graph $\mathcal{T}^i(\mathcal{G},q)$.

\subsection{Data Flow}

This subsection describes the data flow within the F-DASF algorithm, which is similar to the one of the DASF framework in \cite{musluoglu2022unifiedp1}, yet is included here for self-containedness.

\begin{figure}[t]
  \includegraphics[width=0.48\textwidth]{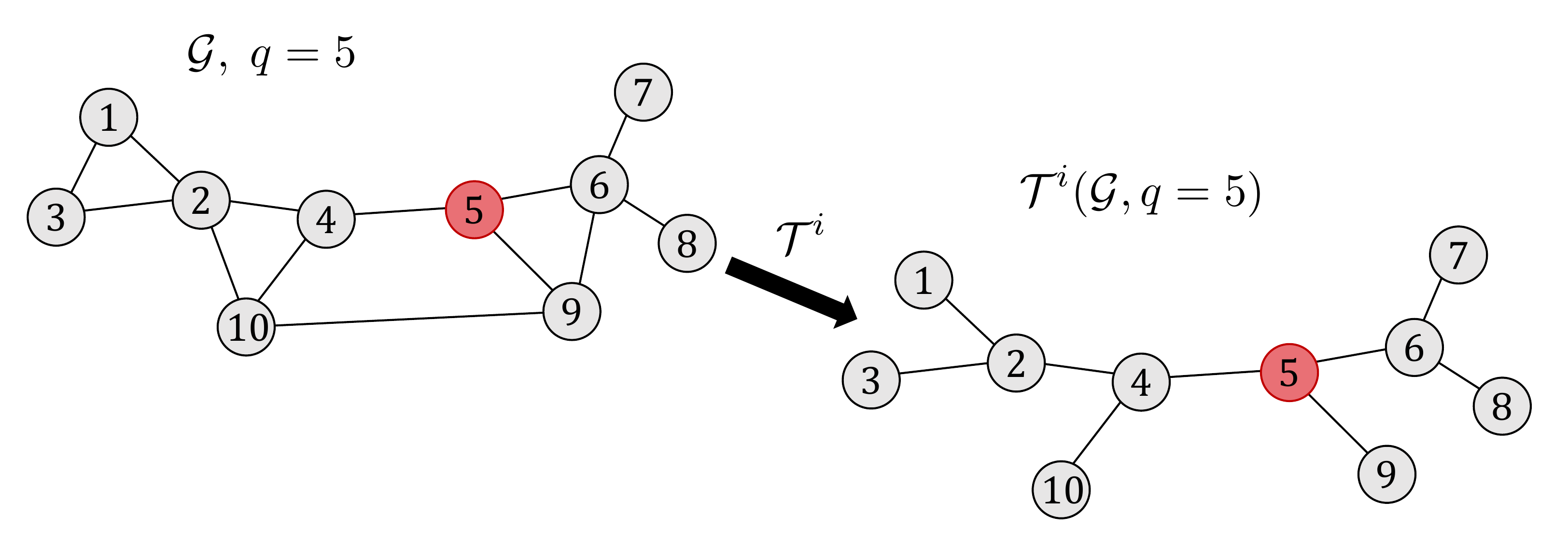}
  \caption{Example of a pruning when the updating node is $q=5$.}
  \label{fig:pruning}
\end{figure}

Let us define $X^i$ to be the estimation of the global filter $X$ at iteration $i$ and partitioned as
\begin{equation}\label{eq:X_part}
  X^i=[X_1^{iT},\dots,X_K^{iT}]^T,
\end{equation}
where $X_k$ is a $M_k\times Q$ matrix such that $X^{iT}\mathbf{y}=\sum_{k}X_k^{iT}\mathbf{y}_k$ and $X^{iT}B=\sum_kX_k^{iT}B_k$. At the beginning of each iteration $i$, every node $k\in\mathcal{K}\backslash\{q\}$ uses its current estimate $X_k^i$ to compress its $M_k-$channel sensor signal $\mathbf{y}_k$ into a $Q-$channel signal, assuming $Q\leq M_k$. A similar operation is done on $B_k$ to obtain
\begin{equation}\label{eq:y_B_compress}
  \widehat{\mathbf{y}}_k^i\triangleq X_k^{iT}\mathbf{y}_k,\;\widehat{B}_k^i\triangleq X_k^{iT}B_k.
\end{equation}
Note that given $X^i$, the objective and constraints of Problems (\ref{eq:prob_g}) and (\ref{eq:aux_prob_g}) can be evaluated at $X=X^i$ at a certain node if that node has access to the sums
\begin{align}\label{eq:X_i_y}
  X^{iT}\mathbf{y}=\sum_{k\in\mathcal{K}}\widehat{\mathbf{y}}_k^i,\;X^{iT}B=\sum_{k\in\mathcal{K}}\widehat{B}_k^i.
\end{align}
The idea is therefore to reconstruct these sums at the updating node $q$ by fusing and forwarding the signals $\widehat{\mathbf{y}}_k^i$ and matrices $\widehat{B}_k^i$ towards node $q$. Note that for the signals $\widehat{\mathbf{y}}_k^i$, this means that each node $k$ will need to transmit $N$ samples\footnote{If $Q>M_k$, node $k$ should simply transmit $N$ samples of its raw signal $\mathbf{y}_k$ to one of its neighbors, say $n\in\mathcal{N}_k$, who treats $\mathbf{y}_k$ as part of its own sensor signal. In this way, node $k$ does not actively participate in the algorithm but acts as an additional set of $M_k$ sensors of node $n$.}, with $N$ large enough to be able to accurately estimate the statistics required to evaluate the objective and constraints of (\ref{eq:prob_g}) and (\ref{eq:aux_prob_g}), implying an $\mathcal{O}(NQ)$ per-node communication cost (assuming $N$ is much larger than the number of columns of $B$).

The way the data is fused and forwarded towards node $q$ is as follows. Each node $k\neq q$ waits until it receives data from all its neighboring nodes $l$ except one, say node $n$. Upon receiving this data, node $k$ sums all the data received from its neighbors to its own compressed data (\ref{eq:y_B_compress}) and transmits this to node $n$. The data transmitted to node $n$ from node $k$ is therefore $N$ samples of the signal
\begin{equation}\label{eq:sum_fwd}
  \widehat{\mathbf{y}}_{k \rightarrow n}^i= X_k^{iT}\mathbf{y}_k+\sum_{l\in\mathcal{N}_k\backslash\{n\}}\widehat{\mathbf{y}}_{l\rightarrow k}^i.
\end{equation}
Note that the second term of (\ref{eq:sum_fwd}) is recursive, and vanishes for leaf nodes, i.e., nodes with a single neighbor. Therefore, the data transmission is initiated at leaf nodes of the pruned network $\mathcal{T}^{i}(\mathcal{G},q)$. The data eventually arrives at the updating node $q$, which receives $N$ samples of the signal
\begin{equation}\label{eq:sum_fwd_n}
  \widehat{\mathbf{y}}_{n\rightarrow q}^i=X_n^{iT}\mathbf{y}_n+\sum_{k\in\mathcal{N}_n\backslash\{q\}}\widehat{\mathbf{y}}_{k\rightarrow n}^i=\sum_{k\in\mathcal{B}_{nq}}\widehat{\mathbf{y}}_k^i
\end{equation}
from all its neighbors $n\in\mathcal{N}_q$. The set $\mathcal{B}_{nq}$ in (\ref{eq:sum_fwd_n}) contains the nodes in the subgraph that includes node $n$ after cutting the edge between nodes $n$ and $q$ in $\mathcal{T}^i(\mathcal{G},q)$. An example of this data flow towards node $q$ is provided in Figure \ref{fig:tree_diagram}. A similar procedure is applied for the deterministic matrix $B$, such that node $q$ receives
\begin{equation}\label{eq:sum_fwd_n_B}
  \widehat{B}_{n\rightarrow q}^i=X_n^{iT}B_n+\sum_{k\in\mathcal{N}_n\backslash\{q\}}\widehat{B}_{k\rightarrow n}^i=\sum_{k\in\mathcal{B}_{nq}}\widehat{B}_k^i,
\end{equation}
from $n\in\mathcal{N}_q$.

\begin{figure}[t]
  \includegraphics[width=0.48\textwidth]{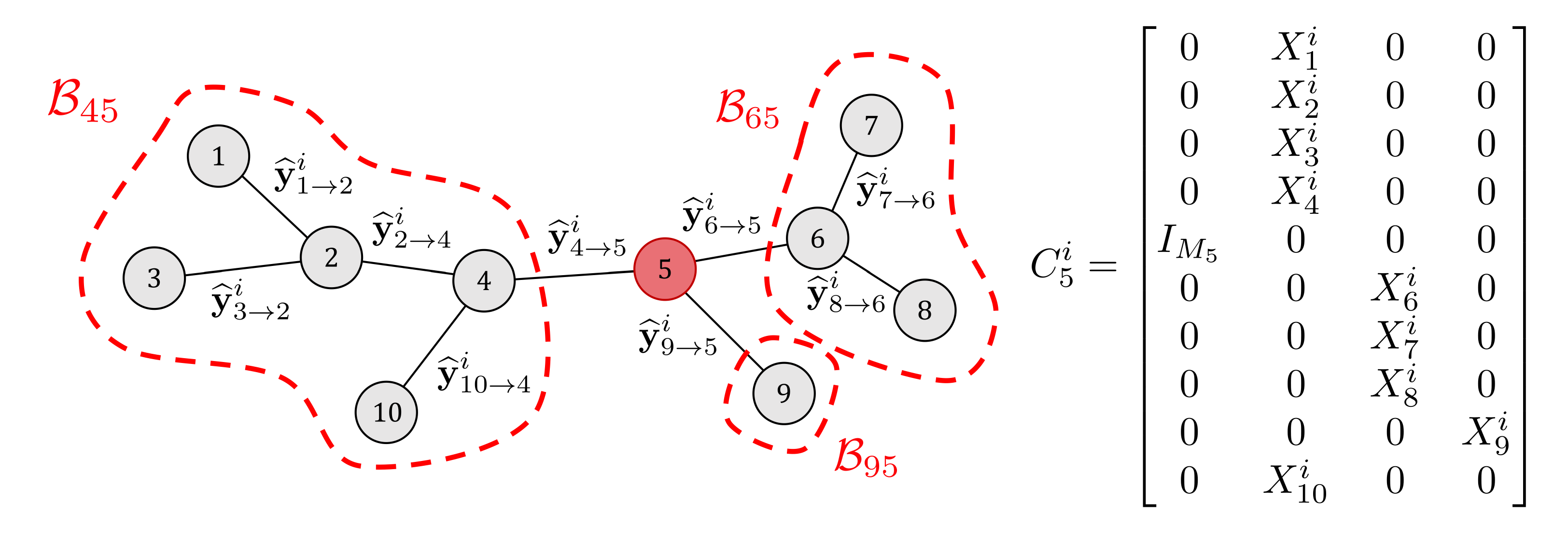}
  \caption{(Adopted from \cite{musluoglu2021distributed}) Example of a tree network where the updating node is node $5$. The clusters containing the nodes ``hidden'' from node $5$ are shown here as $\mathcal{B}_{45}$, $\mathcal{B}_{65}$, $\mathcal{B}_{95}$. The resulting transition matrix is given by $C_5^i$.}
  \label{fig:tree_diagram}
\end{figure}

Let us now label the neighboring nodes of node $q$ as $\mathcal{N}_q\triangleq\{n_1,\dots,n_{|\mathcal{N}_q|}\}$. Node $q$'s own observation $\mathbf{y}_q$ and the compressed signals obtained from its neighbors can then be concatenated as
\begin{equation}\label{eq:tree_data}
    \widetilde{\mathbf{y}}_q^i\triangleq[\mathbf{y}_q^T,\widehat{\mathbf{y}}_{n_1\rightarrow q}^{iT},\dots,\widehat{\mathbf{y}}_{n_{|\mathcal{N}_q|}\rightarrow q}^{iT}]^T\;\in\mathbb{R}^{\widetilde{M}_q},
\end{equation}
where $\widetilde{M}_q=|\mathcal{N}_q|\cdot Q+M_q$, and similarly for the deterministic term $B$:
\begin{equation}\label{eq:B_tilde}
  \widetilde{B}_q^i\triangleq [B_q^T,\widehat{B}_{n_1\rightarrow q}^{iT},\dots,\widehat{B}_{n_{|\mathcal{N}_q|}\rightarrow q}^{iT}]^T\;\in\mathbb{R}^{\widetilde{M}_q\times L}.
\end{equation}
The quantities $\widetilde{\mathbf{y}}_q^i$ and $\widetilde{B}_q^i$ represent the data available at the updating node $q$ during iteration $i$. In the original DASF algorithm \cite{musluoglu2022unifiedp1}, these local data are then used to construct a local (compressed) version of the original global problem at node $q$. Defining a new variable $\widetilde{X}_q\in\mathbb{R}^{\widetilde{M}_q\times Q}$ at node $q$ intending to mimic the global variable $X$ locally, node $q$ creates a local version of the global problem (\ref{eq:prob_g}), given by
\begin{equation}\label{eq:loc_prob}
  \begin{aligned}
    \underset{\widetilde{X}_q\in\mathbb{R}^{\widetilde{M}_q\times Q}}{\text{minimize } } & \frac{\varphi_1\left(\widetilde{X}_q^T\widetilde{\mathbf{y}}_q^i(t),\widetilde{X}_q^T\widetilde{B}_q^i\right)}{\varphi_2\left(\widetilde{X}_q^T\widetilde{\mathbf{y}}_q^i(t),\widetilde{X}_q^T\widetilde{B}_q^i\right)}\\
  \textrm{subject to } & \eta_j\left(\widetilde{X}_q^T\widetilde{\mathbf{y}}_q^i(t),\widetilde{X}_q^T\widetilde{B}_q^i\right)\leq 0\;\textrm{ $\forall j\in\mathcal{J}_I$},\\
    & \eta_j\left(\widetilde{X}_q^T\widetilde{\mathbf{y}}_q^i(t),\widetilde{X}_q^T\widetilde{B}_q^i\right)=0\;\textrm{ $\forall j\in\mathcal{J}_E$}.
  \end{aligned}
\end{equation}
At this point, the updating node $q$ can follow the steps of the DASF algorithm \cite{musluoglu2022unifiedp1} and locally solve Problem (\ref{eq:loc_prob}) to obtain the optimal local value $\widetilde{X}_q^*$, after which a new DASF iteration starts with a new updating node. Note that since (\ref{eq:loc_prob}) is a fractional program\footnote{This is a general property of the DASF framework: the local problems to be solved at each node exhibit the same structure as the centralized problem and can therefore use the same solver.}, node $q$ uses the Dinkelbach procedure provided in Algorithm \ref{alg:dinkelbach} to solve it. However, since the Dinkelbach procedure is itself iterative, using it to solve (\ref{eq:loc_prob}) is computationally expensive, creating nested iterations inside the outer-loop iterations of the DASF algorithm. Therefore, for generic fractional programs solved using Algorithm \ref{alg:dinkelbach}, the DASF algorithm operates at two different time scales. In the next subsection, we propose a modification to the DASF algorithm that avoids this problem by only solving a \textit{single} step of the Dinkelbach procedure at each DASF iteration, greatly reducing the computational burden at each node.

\subsection{The Fractional DASF (F-DASF) Algorithm}

Instead of solving a compressed version of the global problem as in (\ref{eq:loc_prob}) at each updating node $q$, the proposed F-DASF algorithm is focused on solving the following compressed version of the auxiliary problem (\ref{eq:loc_aux_prob}):
\begin{equation}\label{eq:loc_aux_prob}
  \begin{aligned}
    \underset{\widetilde{X}_q\in\mathbb{R}^{\widetilde{M}_q\times Q}}{\text{minimize } } & \varphi_1\left(\widetilde{X}_q^T\widetilde{\mathbf{y}}_q^i(t),\widetilde{X}_q^T\widetilde{B}_q^i\right)-\rho^i\varphi_2\left(\widetilde{X}_q^T\widetilde{\mathbf{y}}_q^i(t),\widetilde{X}_q^T\widetilde{B}_q^i\right)\\
  \textrm{subject to } & \eta_j\left(\widetilde{X}_q^T\widetilde{\mathbf{y}}_q^i(t),\widetilde{X}_q^T\widetilde{B}_q^i\right)\leq 0\;\textrm{ $\forall j\in\mathcal{J}_I$},\\
    & \eta_j\left(\widetilde{X}_q^T\widetilde{\mathbf{y}}_q^i(t),\widetilde{X}_q^T\widetilde{B}_q^i\right)=0\;\textrm{ $\forall j\in\mathcal{J}_E$},
  \end{aligned}
\end{equation}
where the term $\rho^i$ appearing in the objective function corresponds to the current value of $r$ for $X^i$, i.e., $\rho^i=r(X^i)$ and is computed at node $q$ as follows. Let us define $\widetilde{X}_q^i$ as
\begin{equation}\label{eq:X_fixed}
  \widetilde{X}_q^i\triangleq[X_q^{iT},I_Q,\dots,I_Q]^T.
\end{equation}
From the definitions provided in (\ref{eq:X_i_y})-(\ref{eq:B_tilde}), we see that $\widetilde{X}_q^i$ is the local equivalent of $X^i$, such that
\begin{equation}\label{eq:filtered_data}
  X^{iT}\mathbf{y}=\widetilde{X}_q^{iT}\widetilde{\mathbf{y}}_q^i,\;X^{iT}B=\widetilde{X}_q^{iT}\widetilde{B}_q^i.
\end{equation}
Since node $q$ has access to $X_q^i$, $\widetilde{\mathbf{y}}_q^i$ and $\widetilde{B}_q^i$, the above equations imply that:
\begin{equation}\label{eq:compute_rho}
  \rho^i=\varrho\left(\widetilde{X}_q^{iT}\widetilde{\mathbf{y}}_q^i(t),\widetilde{X}_q^{iT}\widetilde{B}_q^i\right)=\frac{\varphi_1\big(\widetilde{X}_q^{iT}\widetilde{\mathbf{y}}_q^i(t),\widetilde{X}_q^{iT}\widetilde{B}_q^i\big)}{\varphi_2\big(\widetilde{X}_q^{iT}\widetilde{\mathbf{y}}_q^i(t),\widetilde{X}_q^{iT}\widetilde{B}_q^i\big)}.
\end{equation}

Note that solving (\ref{eq:loc_aux_prob}) corresponds to executing a single iteration of Dinkelbach's procedure in Algorithm \ref{alg:dinkelbach}. In the case of the standard DASF algorithm, solving the corresponding local problem (\ref{eq:loc_prob}) would require solving multiple instances of (\ref{eq:loc_aux_prob}), each time with a different value for $\rho$, in order to execute all iterations in Algorithm \ref{alg:dinkelbach} until convergence.

\begin{figure}[!t]
  \removelatexerror
  \DontPrintSemicolon
  \begin{algorithm}[H]
  \caption{F-DASF Algorithm}\label{alg:f_dasf}
  \SetKwInOut{Output}{output}
  \BlankLine
  $X^0$ initialized randomly, $i\gets0$.\;
  \Repeat
  {
  Choose the updating node as $q\gets (i\mod K)+1$.\;
  1) The network $\mathcal{G}$ is pruned into a tree $\mathcal{T}^i(\mathcal{G},q)$.\;

  2) Every node $k$ collects $N$ samples of $\mathbf{y}_k$. All nodes compress these to $N$ samples of $\widehat{\mathbf{y}}^{i}_k$ and also compute $\widehat{B}_k^i$ as in (\ref{eq:y_B_compress}).\;

  3) The nodes sum-and-forward their compressed data towards node $q$ via the recursive rule (\ref{eq:sum_fwd}) (and a similar rule for the $\widehat{B}_k^i$'s). Node $q$ eventually receives $N$ samples of $\widehat{\mathbf{y}}^{i}_{n\rightarrow q}$ given in (\ref{eq:sum_fwd_n}), and the matrix $\widehat{B}_{n\rightarrow q}^i$ defined in (\ref{eq:sum_fwd_n_B}), from all its neighbors $n\in\mathcal{N}_q$.\; 

  \At{Node $q$}
  {
    4a) Compute $\rho^i$ as in (\ref{eq:compute_rho}).\;
    4b) Compute (\ref{eq:compute_Xtilde}), i.e., execute a single Dinkelbach iteration by solving (\ref{eq:loc_aux_prob}), resulting in $\widetilde{X}_q^{*}$. If the solution is not unique, select a solution via (\ref{eq:select_sol}).\;
    4c) Partition $\widetilde{X}^{*}_q$ as in (\ref{eq:X_tilde_part}).\;
    4d) Disseminate every $G_n^{i+1}$ in the corresponding subgraph $\mathcal{B}_{nq}$.\;
  }
  
  5) Every node updates $X_k^{i+1}$ according to (\ref{eq:upd_tree}).\;
  
  $i\gets i+1$\;
  }
  \end{algorithm}
\end{figure}

We define $\widetilde{X}_q^{*}$ as the solution of (\ref{eq:loc_aux_prob}), i.e.,
\begin{equation}\label{eq:compute_Xtilde}
\resizebox{.48\textwidth}{!}{%
    $\widetilde{X}_q^{*}\triangleq\underset{\widetilde{X}_q\in\widetilde{\mathcal{S}}_q^i}{\text{argmin }}\varphi_1 \left(\widetilde{X}_q^T\widetilde{\mathbf{y}}_q^i(t),\widetilde{X}_q^T\widetilde{B}_q^i\right)-\rho^i\varphi_2\left(\widetilde{X}_q^T\widetilde{\mathbf{y}}_q^i(t),\widetilde{X}_q^T\widetilde{B}_q^i\right),$%
    }
\end{equation}
where $\widetilde{\mathcal{S}}_q^i$ is the constraint set of (\ref{eq:loc_aux_prob}). If (\ref{eq:loc_aux_prob}) does not have a unique solution, $\widetilde{X}_q^{*}$ is selected as the solution closest to $\widetilde{X}_q^{i}$ in (\ref{eq:X_fixed}), i.e.,
\begin{equation}\label{eq:select_sol}
  \widetilde{X}_q^{*}=\underset{\widetilde{X}_q\in\widetilde{\mathcal{X}}_q^i}{\text{argmin }}||\widetilde{X}_q-\widetilde{X}_q^i||_F
\end{equation}
where $\widetilde{\mathcal{X}}_q^i$ is the set of possible solutions $\widetilde{X}_q$ of (\ref{eq:loc_aux_prob}) and $\widetilde{X}_q^i$ is given in (\ref{eq:X_fixed}). Note that the Frobenius norm to select $\widetilde{X}_q^{*}$ in (\ref{eq:select_sol}) is arbitrary, and another distance function can be selected as long as it is continuous.

Finally, let us partition $\widetilde{X}_q^{*}$ as
\begin{equation}\label{eq:X_tilde_part}
  \widetilde{X}_q^{*}=[X_q^{(i+1)T},G_{n_1}^{(i+1)T},\dots,G_{n_{|\mathcal{N}_q|}}^{(i+1)T}]^T,
\end{equation}
such that $G_n$ is $Q\times Q$, $\forall n\in\mathcal{N}_q$, where the $G-$matrices consist of the part of $\widetilde{X}_q^{*}$ that is multiplied with the compressed signals of the neighbors of node $q$ in the inner product $\widetilde{X}_q^{*T}\widetilde{\mathbf{y}}_q^i$. Every matrix $G_{n}^{i+1}$ is then disseminated in the pruned network to the corresponding subgraph $\mathcal{B}_{nq}$ through node $n\in\mathcal{N}_q$, and every node $k\in\mathcal{K}$ can update its local variable $X_k$ as
\begin{equation}\label{eq:upd_tree}
  X_k^{i+1}=\begin{cases}
  X_q^{i+1} & \text{if $k=q$} \\
  X_k^{i}G_n^{i+1} & \text{if $k\in\mathcal{B}_{nq}$, $n\in\mathcal{N}_q$}.
  \end{cases}
\end{equation}
This process is then repeated by selecting different updating nodes at different iterations. Note that every neighbor $n\in\mathcal{N}_q$ of the updating node $q$ has a corresponding matrix $G_n$. As mentioned in Section \ref{sec:pruning}, if the pruning function $\mathcal{T}^i$ were to remove a link between an updating node and one (or multiple) of its neighbors, the convergence speed of the F-DASF algorithm would suffer from it. Indeed, this would translate in less degrees of freedom in the local problem (\ref{eq:loc_aux_prob}) itself, in turn leading to a fewer number of $G_n$ matrices used for updating the global variable $X$ as in (\ref{eq:upd_tree}).

Algorithm \ref{alg:f_dasf} summarizes the steps of the proposed F-DASF algorithm presented in this section. It is noted that at each iteration $i$, different batches of $N$ samples of signals $\mathbf{y}_k$ are used. This means that the iterations of the F-DASF algorithm are spread out over different signal windows across time, making it able to track changes in the signal statistics, similar to an adaptive filter (see Section \ref{sec:adaptive} for an example). A visual comparison of the difference between the DASF and F-DASF algorithms is provided in Figure \ref{fig:comparison}. The figure highlights how the F-DASF algorithm is computationally much more attractive compared to the original DASF method for fractional programs.

\begin{figure}[t]
  \includegraphics[width=0.48\textwidth]{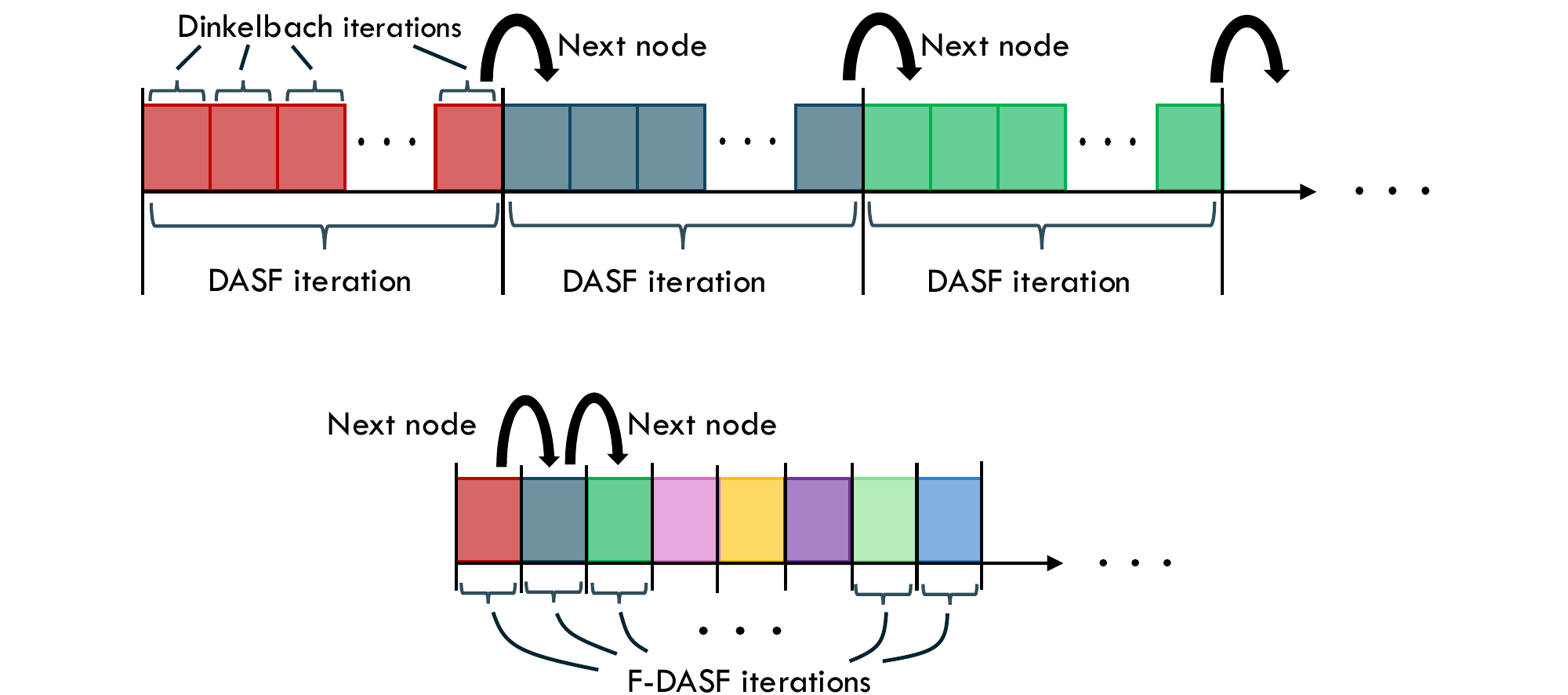}
  \caption{Visual comparison between the DASF (top) and F-DASF (bottom) algorithms. Each rectangle corresponds to one Dinkelbach iteration in both scenarios.}
  \label{fig:comparison}
\end{figure}

In Section \ref{sec:convergence}, we will show that the F-DASF algorithm converges to the solution of the centralized problem (\ref{eq:prob_g}), despite the fact that none of the nodes have access to the full signal $\mathbf{y}$. In fact, each node only has access to the fused data received from its neighbors, such that the full set of second-order statistics of $\mathbf{y}$ is never observable. While these statements are also valid for the DASF algorithm, the F-DASF method achieves the same result with a significantly reduced computational cost.

\begin{rem}\label{rem:f_dasf_interpretation}
  It is worth mentioning that although we have presented the F-DASF algorithm as applying a single iteration of Algorithm \ref{alg:dinkelbach} within a DASF outer loop, F-DASF can also be interpreted as applying a single iteration of the DASF algorithm \cite{musluoglu2022unifiedp1} in step $1$ of Algorithm \ref{alg:dinkelbach}. In the latter case, instead of solving the auxiliary problem (\ref{eq:aux_prob_g}) in a distributed fashion by fixing $\rho$ and applying DASF until convergence, we would only partially solve it by only performing a single iteration of DASF, i.e., a node solves its local problem and immediately updates $\rho$, which would again result in the F-DASF algorithm presented in this section.
\end{rem}

\subsection{Practical Considerations}

Although the F-DASF algorithm can directly be implemented in practice by following the steps presented in Algorithm \ref{alg:f_dasf}, we discuss here some practical considerations to take into account.
\subsubsection{Non-stationarity} When using real data, the assumptions on the statistical properties of the signals discussed in Section \ref{sec:adaptivity} might not always hold in practice. For example, it is well-known that electroencephalography (EEG) signals are highly non-stationary \cite{Siuly_2016}, yet fractional programming problems are commonplace in EEG signal processing (for example in the form of the well-known common spatial patterns algorithm \cite{blankertz2007optimizing,wang2007trace}). In this case, the resulting filters fit to the longer term average statistics of the signals. Furthermore, the spatial coherence across different channels is often more stationary than the signals produced by the underlying signal sources, in which case the filters will mainly track and exploit these more stable spatial relationships. For example, in microphone array processing, the speech sources themselves are highly non-stationary, but their spatial location is often fixed or only slowly time-varying. If even these more stable spatial relationships turn out to change too fast between different batches of data, one way to improve the tracking properties of the F-DASF algorithm would be to apply multiple iterations over the same signal set. Note that this would increase the communication cost within the network, as the transmitted signals $\widehat{\mathbf{y}}_{k\rightarrow n}^i$ in (\ref{eq:sum_fwd}) are still updated by the new values of $X_k$ at each iteration. The method described in \cite{musluoglu2022tracking} provides a good tradeoff between the tracking improvement and communication cost increase in such modifications.

\subsubsection{Effect of problem parameters}\label{sec:prob_params} Another important aspect to consider is how the problem setting parameters affect the algorithm performance (we also refer the reader to \cite{musluoglu2022unifiedp1} for further details on this). The number of columns $Q$ of $X$ is directly linked to the amount of compression applied to the raw signals $\mathbf{y}_k$ at each node $k$. A larger ratio $Q/M_k$ implies less compression, i.e., transmission of more informative signals, therefore faster convergence of the F-DASF algorithm, while smaller values of $Q$ lead to smaller communication costs per link at the expense of a slower convergence. Note that in various problems, the value of $Q$ might be imposed by the problem itself, e.g., when the aim is to project the data into a $2-$dimensional space, implying $Q=2$.

The network topology also affects the convergence speed, as more connected graphs result in larger degrees of freedom when updating the $X_k$'s in (\ref{eq:upd_tree}). This is because the average number of neighbors for each updating node will be higher such that fewer nodes will have to share the same $G-$matrices, i.e., the tree rooted in the updating node will have more branches.

As for the number of nodes in the network, the relationship between $K$ and the convergence is less straightforward than the other parameters discussed above. In general, the convergence speed during the first $K$ iterations is expected to be slower for larger $K$'s, as a full round of update where every node becomes the updating node would take longer to complete. This is especially significant if the starting estimate $X^0$ of $X$ is chosen randomly, since the random initialization components will be present until a full round of update is done. For large networks, it is therefore beneficial to start from a value $X^0$ that is fixed, e.g., using prior knowledge of the problem such as an old estimate. Alternatively, the first $K$ iterations of F-DASF could be done in parallel as a faster initialization step, allowing to map the $X_k$'s to the correct scale.

\subsubsection{Link and node failures} The F-DASF algorithm is robust to link failures as long as there are other paths within the graph where any node can reach any other one (note that a new spanning tree is formed at each iteration), and that a technical condition that relates the number of constraints with the network topology (see equation (\ref{eq:J_upper_bound_2}) in Section \ref{sec:convergence}) remains satisfied. This also holds in the case where a link connected to the updating node fails at the expense of reduced convergence speed due to the loss of degrees of freedom when updating the $X_k$'s. On the other hand, if a node fails in the long term, the F-DASF algorithm will (re-)converge to the optimal solution using the data from all remaining nodes, assuming that the rest of the network remains connected.

\subsubsection{Non-ideal Communication Links} In this paper, we have made abstraction of non-idealities in the communication links, i.e., we do not take quantization noise or data loss into account. We also assume a synchronized setting, implying that communication delays accumulate across the network, as nodes would wait until they receive the required data from their neighbors, before fusing and transmitting it to the next node. Furthermore, we assume that all nodes sample at the exact same sampling rate, which in practice requires synchronization protocols to compensate for inter-node sampling rate offsets.

\section{Technical Analysis and Convergence}\label{sec:convergence}

It can be shown that the F-DASF algorithm inherits the same convergence properties as the ones of the DASF algorithm detailed in \cite{musluoglu2022unifiedp2}, i.e., under mild technical constraints often satisfied in practice, the F-DASF algorithm converges to the optimal solution of the global problem (\ref{eq:prob_frac}). We note that these convergence results do not trivially follow from the convergence results of the original DASF algorithm. This is because DASF requires the updating nodes to \textit{fully} solve (\ref{eq:loc_prob}). Instead, the updating nodes in F-DASF only perform a \textit{single} iteration of Dinkelbach's procedure\footnote{Alternatively, as mentioned in Remark \ref{rem:f_dasf_interpretation}, this can also be viewed as a global Dinkelbach procedure in which a \textit{single} DASF iteration is performed at each iteration.}, which corresponds to solving (\ref{eq:loc_aux_prob}). Although some results from the convergence proof of DASF can be re-used, some non-trivial extensions are required to establish convergence of F-DASF. In this subsection, we will therefore concentrate on the results that differ largely from the DASF case, while similar ones will be briefly mentioned for completeness. In particular, in the DASF scenario, proving the monotonic decrease of the objective and the stationarity of fixed points of the algorithm relied on the fact that each node fully solves a compressed version of the global problem. The main challenge for the F-DASF scenario is to still obtain these results despite this difference, by using fractional programming results presented in Section \ref{sec:frac_prog}. Similar to the proof in \cite{musluoglu2022unifiedp2}, for mathematical tractability, we assume that all signals are stationary and that each node is able to perfectly estimate the statistics of its locally available signals, i.e., as if $N\rightarrow+\infty$. This means that the proof is only asymptotically valid, i.e., approximating settings with sufficiently large batch sizes $N$. In practice, estimation errors on these statistics (due to finite $N$) will result in the algorithm only converging to a neighborhood of the true optimal point, where the solution will ``hover'' around this optimal point.

\subsection{Preliminaries}
We start by describing the relationship between the local problems and the global problem which will be useful for the technical analyses that will be presented in the later parts of this section. 
From equations (\ref{eq:sum_fwd_n})-(\ref{eq:B_tilde}), we observe that there exists a matrix $C_q^i\in\mathbb{R}^{M\times \widetilde{M}_q}$ such that
\begin{equation}\label{eq:compress_y_B}
  \widetilde{\mathbf{y}}_q^i=C_q^{iT}\mathbf{y},\; \widetilde{B}_q^i=C_q^{iT}B,
\end{equation}
i.e., there exists a linear (compressive) relationship between the local data at the updating node $q$ and the global data $\mathbf{y}$ and $B$. 
An example of such a matrix $C_q^i$ is given in Figure \ref{fig:tree_diagram} for the network given in the same figure.
Note that, since the first $M_q$ entries in $\widetilde{\mathbf{y}}_q^i$ are equal to $\mathbf{y}_q$ (see (\ref{eq:tree_data})), the matrix $C_q^i$ must have the structure
\begin{equation}\label{eq:cqi_tree}
  C_q^i=\left[
    \begin{array}{c|c}
    0 &  \\
    I_{M_q} & * \\
    0 & 
    \end{array}
    \right],
\end{equation}
where $I_{M_q}$ is placed in the $q-$th block-row corresponding to the position of the entries $\mathbf{y}_q$ in $\mathbf{y}$.

Similar to (\ref{eq:filtered_data}), we observe that when node $q$ applies a generic spatial filter $\widetilde{X}_q\in\mathbb{R}^{\widetilde{M}_q\times Q}$ to $\widetilde{\mathbf{y}}_q^i$, this is equivalent to applying a generic spatial filter $X\in\mathbb{R}^{M\times Q}$ to the network-wide sensor signal $\mathbf{y}$, i.e., $\widetilde{X}_q^T\widetilde{\mathbf{y}}_q^i=X^T\mathbf{y}$. With (\ref{eq:compress_y_B}), we then have $\widetilde{X}_q^{T}\widetilde{\mathbf{y}}_q^i=\widetilde{X}_q^{T}(C_q^{iT}\mathbf{y})=(C_q^i\widetilde{X}_q)^T\mathbf{y}$, while following similar steps for the deterministic matrix $B$ gives $\widetilde{X}_q^{T}\widetilde{B}_q^i=\widetilde{X}_q^{T}(C_q^{iT}B)=(C_q^i\widetilde{X}_q)^TB$, such that
\begin{equation}\label{eq:param_X}
  X=C_q^i\widetilde{X}_q.
\end{equation}
When comparing (\ref{eq:aux_prob_g}) with (\ref{eq:loc_aux_prob}), we observe that their objective and constraint functions are indeed the same if $\widetilde{X}_q^T\widetilde{\mathbf{y}}_q^i=X^T\mathbf{y}$.
This implies that if $\widetilde{X}_q$ is a feasible point of the local auxiliary problem (\ref{eq:loc_aux_prob}), the point $X$ parameterized by (\ref{eq:param_X}) is a feasible point of the global auxiliary problem (\ref{eq:aux_prob_g}), and vice versa, as formalized in the following lemma (the proof is equivalent to the proof provided in \cite{musluoglu2022unifiedp1} for the original DASF algorithm and is therefore omitted here).
\begin{lem}\label{lem:X_in_constraints}
  For any iteration $i>0$ of Algorithm \ref{alg:f_dasf},
  \begin{equation}\label{eq:loc_glob}
    \widetilde{X}_q\in\widetilde{\mathcal{S}}_q^i\iff C_q^i\widetilde{X}_q\in\mathcal{S}.
  \end{equation}
  In particular, $X^i\in\mathcal{S}$ and $\widetilde{X}_q^{i}\in\widetilde{\mathcal{S}}_q^i$ for all $i>0$, where $\widetilde{X}_q^i$ is defined in (\ref{eq:X_fixed}).
\end{lem}
We note that this lemma always holds, independent of whether the initialization point of F-DASF is in $\mathcal{S}$ or not. In particular, the last sentence of the lemma is important as it implies that every output $X^i$ of the F-DASF algorithm is a feasible point of the global auxiliary problem (\ref{eq:aux_prob_g}) and hence also of the global problem (\ref{eq:prob_g}) since they have the same constraint set.

The relationship (\ref{eq:param_X}) also allows us to compactly write Problem (\ref{eq:loc_aux_prob}) using the functions $f_j$ and $h_j$ defined in (\ref{eq:r_f_h}):
\begin{equation}\label{eq:loc_aux_prob_compact}
  \begin{aligned}
      \underset{\widetilde{X}_q\in\mathbb{R}^{\widetilde{M}_q\times Q}}{\text{minimize } } \quad & f(C_q^i\widetilde{X}_q,\rho^i)=f_1\left(C_q^i\widetilde{X}_q\right)-\rho^i f_2\left(C_q^i\widetilde{X}_q\right)\\
    \textrm{subject to } \quad & h_j\left(C_q^i\widetilde{X}_q\right)\leq 0,\;\textrm{ $\forall j\in\mathcal{J}_I$},\\
     & h_j\left(C_q^i\widetilde{X}_q\right)=0,\;\textrm{ $\forall j\in\mathcal{J}_E$},
  \end{aligned}
\end{equation}
where
\begin{equation}\label{eq:compute_rho_f}
  \rho^i=\frac{f_1\big(C_q^i\widetilde{X}_q^i\big)}{f_2\big(C_q^i\widetilde{X}_q^i\big)}.
\end{equation}
In the remaining parts of this section, we will often refer to the compact notation in equations (\ref{eq:prob_frac}), (\ref{eq:aux_prob}) and (\ref{eq:loc_aux_prob_compact}) to refer to the global problem, its auxiliary problems, and the compressed auxiliary problem, respectively (instead of (\ref{eq:prob_g}), (\ref{eq:aux_prob_g}) and (\ref{eq:loc_aux_prob})) for notational convenience.

\subsection{Convergence in Objective}

The following theorem establishes the convergence of the auxiliary parameter $\rho^i$, which then also implies that there is convergence in the objective of (\ref{eq:prob_frac}) / (\ref{eq:prob_g}) via the relationship in (\ref{eq:compute_rho}).

\begin{thm}\label{thm:monotonic}
  The sequence $(\rho^i)_i$ generated by Algorithm \ref{alg:f_dasf} is a non-increasing, converging sequence.
\end{thm}
\begin{proof}
  From (\ref{eq:compute_rho_f}), we have $f\big(C_q^i\widetilde{X}_q^i,\rho^i\big)$ $=f_1\big(C_q^i\widetilde{X}_q^i\big)-\rho^i f_2\big(C_q^i\widetilde{X}_q^i\big)=0$. Since $\widetilde{X}_q^{*}$ solves (\ref{eq:loc_aux_prob_compact}), we have that $f\big(C_q^i\widetilde{X}_q^{*},\rho^i\big)\leq f\big(C_q^i\widetilde{X}_q,\rho^i\big)$ for any $\widetilde{X}_q\in\widetilde{\mathcal{S}}_q^i$. We note that $\widetilde{X}_q^i$ as defined in (\ref{eq:X_fixed}) is a feasible point of (\ref{eq:loc_aux_prob}), i.e., belongs to $\widetilde{\mathcal{S}}_q^i$ as is shown in Lemma \ref{lem:X_in_constraints}. We then write $f\big(C_q^i\widetilde{X}_q^{*},\rho^i\big)\leq f\big(C_q^i\widetilde{X}_q^i,\rho^i\big)=0$. After reordering the terms of $f\big(C_q^i\widetilde{X}_q^{*},\rho^i\big)$, we obtain $\frac{f_1\big(C_q^i\widetilde{X}_q^{*}\big)}{f_2\big(C_q^i\widetilde{X}_q^{*}\big)}$ $=\rho^{i+1}\leq \rho^i$. Therefore, the sequence $(\rho^i)_i$ is monotonic non-increasing and since it is lower bounded by $\rho^*$, it must converge.
\end{proof}

\noindent We note that convergence of $(\rho^i)_i$ does not necessarily imply convergence of the underlying sequence $(X^i)_i$, nor the optimality of the resulting accumulation point, which is the topic of the next two subsections.

\subsection{Stationarity of Fixed Points}
A fixed point of the F-DASF algorithm is defined as a point $\widebar{X}\in\mathcal{S}$ that is invariant under any F-DASF update step (for any updating node $q$), i.e., $X^i=\widebar{X}$, $\forall i>0$ when initializing Algorithm \ref{alg:f_dasf} with $X^0=\widebar{X}$. We will now present results that guarantee that such fixed points of the F-DASF algorithm are stationary points of the global problem (\ref{eq:prob_g}) under mild technical conditions that are akin to the standard LICQ conditions in the optimization literature \cite{peterson1973review}.

\begin{subcond}{cond}\label{subcond:lin_indep}
  \begin{cond}\label{cond:lin_indep_1}
    For a fixed point $\widebar{X}$ of Algorithm \ref{alg:f_dasf}, the elements of the set $\{\widebar{X}^{T}\nabla_X h_j(\widebar{X})\}_{j\in\mathcal{J}}$ are linearly independent.
  \end{cond}

  \noindent Condition \ref{cond:lin_indep_1} requires the linear independence of a set of $J$ matrices of size $Q\times Q$, which can only be satisfied if
  \begin{equation}\label{eq:J_upper_bound_1}
    J\leq Q^2.
  \end{equation}
  Although Condition \ref{cond:lin_indep_1} requires checking linear independence at fixed points, which are usually unknown, the cases where Condition \ref{cond:lin_indep_1} would be violated (while (\ref{eq:J_upper_bound_1}) is satisfied) are very contrived and highly improbable \cite{musluoglu2022unifiedp2}. Additionally, this condition can be verified beforehand for some commonly encountered constraint sets, such as the generalized Stiefel manifold. When Condition \ref{cond:lin_indep_1} is satisfied, we obtain a first result on the stationarity of the fixed points of the F-DASF algorithm.

  \begin{thm}\label{thm:stationarity}
    If Condition \ref{cond:lin_indep_1} is satisfied for a fixed point $\widebar{X}$ of Algorithm \ref{alg:f_dasf}, then $\widebar{X}$ is a stationary point of both (i) the auxiliary problem (\ref{eq:aux_prob}) (or (\ref{eq:aux_prob_g})) for $\rho=r(\widebar{X})$ and (ii) the global problem (\ref{eq:prob_frac}) (or (\ref{eq:prob_g})), satisfying their KKT conditions.
  \end{thm}
  \begin{proof}
      We will first show that a fixed point of the F-DASF algorithm is a KKT point of the auxiliary problem (\ref{eq:aux_prob}). Using this result, we will derive the steps to show that such a point is also a KKT point of the global problem (\ref{eq:prob_frac}).
  
  The KKT conditions of the auxiliary problem (\ref{eq:aux_prob}) can be written as:
  \begin{align}
    &\nabla_X \mathcal{L}_{\rho}(X,\rho,\bm{\lambda})=0,\label{eq:stat}\\
    &h_j(X)\leq 0\;\textrm{ $\forall j\in\mathcal{J}_I$, } h_j(X)=0\;\textrm{ $\forall j\in\mathcal{J}_E$,}\label{eq:pf}\\
    &\lambda_j\geq 0\;\textrm{ $\forall j\in\mathcal{J}_I$,}\label{eq:dual}\\
    &\lambda_j h_j(X)=0\;\textrm{ $\forall j\in\mathcal{J}_I$,}\label{eq:comp_slack}
  \end{align}
  where 
  \begin{equation}\label{eq:lagrangian}
    \mathcal{L}_{\rho}(X,\rho,\bm{\lambda})\triangleq f(X,\rho)+\sum_{j\in\mathcal{J}} \lambda_jh_j(X)
  \end{equation}
  is the Lagrangian of (\ref{eq:aux_prob}). We use $\bm{\lambda}$ in bold as a shorthand notation for the set of all Lagrange multipliers $\lambda_j\in\mathbb{R}$ corresponding to the constraint $h_j$.
  
  At iteration $i$, the updating node $q$ solves the local auxiliary problem (\ref{eq:loc_aux_prob_compact}) which has the following Lagrangian
  \begin{equation}\label{eq:lagrangian_local}
    \widetilde{\mathcal{L}}_{\rho}(\widetilde{X}_q,\rho^i,\widetilde{\bm{\lambda}}(q))\triangleq f(C_q^i\widetilde{X}_q,\rho^i)+\sum_{j\in\mathcal{J}} \lambda_j(q)h_j(C_q^i\widetilde{X}_q),
  \end{equation}
  where the $\lambda_j(q)$'s are the Lagrange multipliers at updating node $q$ and iteration $i$ corresponding to the local problem (\ref{eq:loc_aux_prob_compact}) and $\widetilde{\bm{\lambda}}(q)$ denotes their collection. Since $\widetilde{X}_q^{*}$ obtained in step $4$b of Algorithm \ref{alg:f_dasf} solves the local problem at node $q$ and iteration $i$, it must satisfy the KKT conditions of the local problem. In particular, the stationarity condition is given as
  \begin{equation}\label{eq:optimality_local}
    \nabla_{\widetilde{X}_q}\widetilde{\mathcal{L}}_{\rho}(\widetilde{X}_q^{*},\rho^i,\widetilde{\bm{\lambda}}(q))=0.
  \end{equation}
  From the parameterization $X=C_q^i\widetilde{X}_q$ in (\ref{eq:param_X}), we have
  \begin{equation}\label{eq:lagrangians}
    \widetilde{\mathcal{L}}_{\rho}(\widetilde{X}_q,\rho^i,\widetilde{\bm{\lambda}}(q))=\mathcal{L}_{\rho}(C_q^i\widetilde{X}_q,\rho^i,\widetilde{\bm{\lambda}}(q)),
  \end{equation}
  allowing us to apply the chain rule on (\ref{eq:optimality_local}) to obtain
  \begin{equation}\label{eq:local_lagrange}
    C_q^{iT}\nabla_X\mathcal{L}_{\rho}(C_q^{i}\widetilde{X}_q^{*},\rho^i,\widetilde{\bm{\lambda}}(q))=0.
  \end{equation}
  Using (\ref{eq:lagrangian}), and the parameterized notation in (\ref{eq:loc_aux_prob_compact}), we find the KKT conditions of the local problem:
  \begin{align}
    &C_q^{iT}\nabla_X\Big[f\left(C_q^i\widetilde{X}_q^{*},\rho^i\right)+\sum_{j\in\mathcal{J}}\lambda_j(q)h_j\left(C_q^i\widetilde{X}_q^{*}\right)\Big]=0, \label{eq:stat_d}\\
    &h_j\left(C_q^i\widetilde{X}_q^{*}\right)\leq 0\;\textrm{ $\forall j\in\mathcal{J}_I$, } h_j\left(C_q^i\widetilde{X}_q^{*}\right)=0\;\textrm{ $\forall j\in\mathcal{J}_E$,}  \label{eq:pf_d}\\
    &\lambda_j(q)\geq 0\;\textrm{ $\forall j\in\mathcal{J}_I$,} \label{eq:dual_d}\\
    &\lambda_j(q)h_j\left(C_q^i\widetilde{X}_q^{*}\right)=0\;\textrm{ $\forall j\in\mathcal{J}_I$,} \label{eq:cs_d}
  \end{align}
  satisfied by $X^{i+1}=C_q^i\widetilde{X}_q^{*}$ and its corresponding set of Lagrange multipliers $\widetilde{\bm{\lambda}}(q)$. Our aim is now to show that at a fixed point of the F-DASF algorithm, i.e., a point such that $X^{i+1}=X^i=\widebar{X}$, the KKT conditions (\ref{eq:stat_d})-(\ref{eq:cs_d}) of the local problem are equivalent to the KKT conditions (\ref{eq:stat})-(\ref{eq:comp_slack}) of the global problem.
  
  Let us first look at the local stationarity condition under the fixed point assumption $X^{i+1}=X^i=\widebar{X}$, which allows us to replace $C_q^i\widetilde{X}_q^{*}=X^{i+1}$ with $C_q^i\widetilde{X}_q^i=X^i=\widebar{X}$, also implying $\rho^{i+1}=\rho^i=\widebar{\rho}=r(\widebar{X})$. Making these changes in (\ref{eq:stat_d}) result in
  \begin{equation}\label{eq:stat_equi}
    C_q^{iT}\nabla_X\Big[f(\widebar{X},\widebar{\rho})+\sum_{j\in\mathcal{J}}\lambda_j(q)h_j(\widebar{X})\Big]=0.
  \end{equation}
  From the structure of $C_q^i$ as displayed in (\ref{eq:cqi_tree}), we observe that the first $M_q$ rows of (\ref{eq:stat_equi}) will be equal to
  \begin{equation}\label{eq:Aq_stat}
    \nabla_{X_q}f(\widebar{X},\widebar{\rho})+\sum_{j\in\mathcal{J}}\lambda_j(q)\nabla_{X_q}h_j(\widebar{X})=0.
  \end{equation} 
  Additionally, the fixed point assumption is independent of the updating node $q$, i.e., whichever node $q$ is selected to be the updating node, the expression (\ref{eq:stat_equi}) will be true. Therefore, (\ref{eq:stat_equi}) and hence (\ref{eq:Aq_stat}) simultaneously hold for any node $q$. Stacking the variations of (\ref{eq:Aq_stat}) for each node $q$ then leads to
  \begin{equation}\label{eq:stat_mat}
    \resizebox{.45\textwidth}{!}{%
    $\begin{bmatrix}
      \nabla_{X_1}f(\widebar{X},\widebar{\rho})\\
      \vdots\\
      \nabla_{X_K}f(\widebar{X},\widebar{\rho})
    \end{bmatrix}=\nabla_Xf(\widebar{X},\widebar{\rho})=-\begin{bmatrix}
      \sum_{j\in\mathcal{J}}\lambda_j(1)\nabla_{X_1}h_j(\widebar{X})\\
      \vdots\\
      \sum_{j\in\mathcal{J}}\lambda_j(K)\nabla_{X_K}h_j(\widebar{X})
    \end{bmatrix}.$%
    }
  \end{equation}
  
  Let us now return to (\ref{eq:stat_equi}) and multiply it by $\widetilde{X}_q^{iT}$ from the left, where $\widetilde{X}_q^i$ is defined in (\ref{eq:X_fixed}). From the fact that $C_q^i\widetilde{X}_q^i=X^i=\widebar{X}$, we can write
  \begin{equation}\label{eq:lin_indep_1}
    \widebar{X}^{T}\nabla_X f(\widebar{X},\widebar{\rho})=-\sum_{j\in\mathcal{J}}\lambda_j(q)\widebar{X}^{T}\nabla_X h_j(\widebar{X}).
  \end{equation}
  From the linear independence of the set $\{\widebar{X}^{T}\nabla_X h_j(\widebar{X})\}_j$ implied by Condition \ref{cond:lin_indep_1}, we obtain the result that the Lagrange multipliers $\{\lambda_j(q)\}_j$ satisfying (\ref{eq:lin_indep_1}) are unique. Note that the left-hand side of (\ref{eq:lin_indep_1}) does not depend on the node $q$, therefore the multipliers are the same for every node, i.e., $\lambda_j(q)=\lambda_j$ for any node $q$. This result can be used to re-write (\ref{eq:stat_mat}) as
  \begin{equation}\label{eq:glob_stat_equi}
    \nabla_X f(\widebar{X},\widebar{\rho})=-\sum_{j\in\mathcal{J}}\lambda_j \nabla_X h_j(\widebar{X}).
  \end{equation}
  Therefore, a fixed point $\widebar{X}$ satisfies the global stationarity condition (\ref{eq:stat}) of the auxiliary problem (\ref{eq:aux_prob}) with corresponding Lagrange multipliers $\{\lambda_j\}_j$.
  
  We now look at the three other KKT conditions. Note that the local primal feasibility condition (\ref{eq:pf_d}) is satisfied globally, since (\ref{eq:pf_d}) and (\ref{eq:pf}) are the same, as was already shown in (\ref{eq:loc_glob}) for any point, so it must also hold for a fixed point. Additionally, $(\widebar{X},\{\lambda_j\}_j)$ satisfies the local versions of both the dual feasibility (\ref{eq:dual_d}) and the complementary slackness condition (\ref{eq:cs_d}), where we replace $C_q^i\widetilde{X}_q^{*}=X^{i+1}$ by $\widebar{X}$ from the fixed point assumption, and $\lambda_j(q)$'s by $\lambda_j$'s for all $j\in\mathcal{J}_I$. Therefore, $(\widebar{X},\{\lambda_j\}_j)$ also satisfies their global counterparts (\ref{eq:dual}) and (\ref{eq:comp_slack}). The pair $(\widebar{X},\{\lambda_j\}_j)$ hence satisfies all the KKT conditions of the auxiliary problem (\ref{eq:aux_prob}).
  
  We now proceed with proving that $\widebar{X}$ also satisfies the KKT conditions of the original fractional program (\ref{eq:prob_frac}), of which the KKT conditions are given by
  \begin{align}
    &\nabla_X \mathcal{L}(X,\bm{\mu})=0,\label{eq:stat_r}\\
    &h_j(X)\leq 0\;\textrm{ $\forall j\in\mathcal{J}_I$, } h_j(X)=0\;\textrm{ $\forall j\in\mathcal{J}_E$,}\label{eq:pf_r}\\
    &\mu_j\geq 0\;\textrm{ $\forall j\in\mathcal{J}_I$,}\label{eq:dual_r}\\
    &\mu_j h_j(X)=0\;\textrm{ $\forall j\in\mathcal{J}_I$,}\label{eq:comp_slack_r}
  \end{align}
  where
  \begin{equation}
    \mathcal{L}(X,\bm{\mu})=r(X)+\sum_{j\in\mathcal{J}}\mu_j h_j(X)
  \end{equation}
  is the Lagrangian of Problem (\ref{eq:prob_frac}). The gradient of $r$ with respect to $X$ can be shown to be equal to
  \begin{equation}
    \nabla_X r(X)=\frac{1}{f_2(X)}\left(\nabla_X f_1(X)-r(X)\nabla_X f_2(X)\right).
  \end{equation}
  By plugging in $\widebar{X}$ and using the fact that $r(\widebar{X})=\widebar{\rho}$, we obtain
  \begin{equation}
    \nabla_X r(\widebar{X})=\frac{1}{f_2(\widebar{X})}\nabla_X f(\widebar{X},\widebar{\rho}),
  \end{equation}
  and by substituting (\ref{eq:glob_stat_equi}) in this equation, we eventually find 
  \begin{equation}
    \nabla_X r(\widebar{X})=-\sum_{j\in\mathcal{J}}\frac{\lambda_j}{f_2(\widebar{X})} \nabla_X h_j(\widebar{X}).
  \end{equation}
  Then, taking $\mu_j\triangleq\lambda_j/f_2(\widebar{X})$ for every $j\in\mathcal{J}$, we observe that we satisfy the stationarity condition (\ref{eq:stat_r}). Additionally, the primal feasibility condition (\ref{eq:pf_r}) is automatically satisfied for $\widebar{X}$ as it is identical to (\ref{eq:pf}). Finally, since $f_2(X)>0$ for any $X\in\mathcal{S}$, the multipliers $\mu_j=\lambda_j/f_2(\widebar{X})$ satisfy (\ref{eq:dual_r}) and (\ref{eq:comp_slack_r}) from the fact that $\lambda_j$'s satisfy (\ref{eq:dual}) and (\ref{eq:comp_slack}). The pair $(\widebar{X},\{\mu_j\}_j)$ therefore satisfies the KKT conditions of the global problem (\ref{eq:prob_frac}).
  \end{proof}

    For the multi-input single-output (MISO) case, the filter $X$ is in fact a vector such that $Q=1$. In this case, (\ref{eq:J_upper_bound_1}) implies that Condition \ref{cond:lin_indep_1} can only be satisfied if we have at most one constraint in Problem (\ref{eq:prob_g}). Similar to \cite{musluoglu2022unifiedp2}, we propose an alternative condition for such cases in order to relax the upper bound on the number of constraints.
    
    \begin{cond}\label{cond:lin_indep_2}
      For a fixed point $\widebar{X}$ of Algorithm \ref{alg:f_dasf}, the elements of the set $\{D_{j,q}(\widebar{X})\}_{j\in\mathcal{J}}$ are linearly independent for any $q$, where
      \begin{equation}\label{eq:Dqj}
        D_{j,q}(\widebar{X})\triangleq \begin{bmatrix}
          \widebar{X}_q^{T}\nabla_{X_q} h_j(\widebar{X})\\
          \sum\limits_{k\in\mathcal{B}_{n_1q}}\widebar{X}_k^{T}\nabla_{X_k} h_j(\widebar{X})\\
          \vdots\\
          \sum\limits_{k\in\mathcal{B}_{n_{|\mathcal{N}_q|}q}}\widebar{X}_k^{T}\nabla_{X_k} h_j(\widebar{X})
        \end{bmatrix},
      \end{equation}
      which is a block-matrix containing $(1+|\mathcal{N}_q|)$ blocks of $Q\times Q$ matrices.
      \end{cond}
\end{subcond}
\renewcommand*{\theHcond}{\thecond}

The size of the matrices $D_{j,q}$ depends on the number of neighbors of node $q$, therefore Condition \ref{cond:lin_indep_2} depends on the topology of the network. In order to satisfy the linear independence of the matrices (\ref{eq:Dqj}), an upper bound analogous to the one in (\ref{eq:J_upper_bound_1}) for Condition \ref{cond:lin_indep_1} can be found to be
\begin{equation}\label{eq:J_upper_bound_2_1}
  J\leq (1+\min_{k\in\mathcal{K}}|\mathcal{N}_k|)Q^2,
\end{equation}
where it is assumed that the pruning function $\mathcal{T}^i(\cdot,q)$ preserves all the links of the updating node $q$. Additionally, we can show that the number of constraints $J$ also needs to satisfy
\begin{equation}\label{eq:J_upper_bound_2_2}
  J\leq \frac{Q^2}{K-1}\sum_{k\in\mathcal{K}}|\mathcal{N}_k|
\end{equation}
for Condition \ref{cond:lin_indep_2} to be satisfied. This second upper bound is related to specific interdependencies between the $D_{j,q}$'s across different nodes $q$, and has been derived in \cite{musluoglu2022unifiedp2}. Combining both (\ref{eq:J_upper_bound_2_1}) and (\ref{eq:J_upper_bound_2_2}), we have
\begin{equation}\label{eq:J_upper_bound_2}
  J\leq \min\left(\frac{Q^2}{K-1}\sum_{k\in\mathcal{K}}|\mathcal{N}_k|,\;(1+\min_{k\in\mathcal{K}}|\mathcal{N}_k|)Q^2\right).
\end{equation}
This is a more relaxed bound than the one in (\ref{eq:J_upper_bound_1}), yet it requires to know the full topology of the network. As in the case of the previous condition, Condition \ref{cond:lin_indep_2} is merely a technical condition, as it is highly likely to be satisfied in practice when (\ref{eq:J_upper_bound_2}) holds. Condition \ref{cond:lin_indep_2} leads to a result analogous to Theorem \ref{thm:stationarity}, given below.
\begin{thm}\label{thm:stationarity_2}
  If Condition \ref{cond:lin_indep_2} is satisfied for a fixed point $\widebar{X}$ of Algorithm \ref{alg:f_dasf}, then $\widebar{X}$ is a stationary point of both (i) the auxiliary problem (\ref{eq:aux_prob}) (or (\ref{eq:aux_prob_g})) for $\rho=r(\widebar{X})$ and (ii) the global problem (\ref{eq:prob_frac}) (or (\ref{eq:prob_g})), satisfying their KKT conditions.
\end{thm}
\noindent The proof of this theorem can be obtained by combining the proof of Theorem \ref{thm:stationarity} provided earlier and the proof of an analogous result for the DASF algorithm \cite[Appendix B]{musluoglu2022unifiedp2}, which is why we omit it here. We note that the bound (\ref{eq:J_upper_bound_2}) should only be satisfied \textit{before} pruning the network using $\mathcal{T}^i$. As long as the pruning function $\mathcal{T}^i$ preserves all links between the updating node $q$ and its neighbors, it can be shown that the result from Theorem \ref{thm:stationarity_2} holds \cite{musluoglu2022unifiedp2}. This is also the case where the topology of the network would change dynamically across iterations, such as when link failures happen, as long as (\ref{eq:J_upper_bound_2}) is still satisfied during any of these failures.

Theorems \ref{thm:stationarity} and \ref{thm:stationarity_2} only guarantee that a fixed point of Algorithm \ref{alg:f_dasf} is a stationary point of Problem (\ref{eq:prob_frac}). A stronger result can be obtained if the fixed point minimizes the auxiliary problem (\ref{eq:aux_prob}).

\begin{thm}(see \cite{jagannathan1966some})\label{thm:global_min}
  If a point $\widebar{X}$ minimizes the auxiliary problem (\ref{eq:aux_prob}) for $\rho=r(\widebar{X})$, then it is a global minimizer of Problem (\ref{eq:prob_frac}).
\end{thm}
\noindent A proof of this theorem can be found in \cite{jagannathan1966some}. Combining Theorems \ref{thm:stationarity}, \ref{thm:stationarity_2}, and \ref{thm:global_min} we can obtain the following result.
\begin{cor}\label{cor:fixed_min}
  Let $\widebar{X}$ be a fixed point of Algorithm \ref{alg:f_dasf} satisfying either Condition \ref{cond:lin_indep_1} or \ref{cond:lin_indep_2}. Then, for $\rho=r(\widebar{X})$, if the auxiliary problem (\ref{eq:aux_prob}) has a unique minimum and no other KKT points, we have $\widebar{X}=X^*$, where $X^*$ is a solution of (\ref{eq:prob_frac}).
\end{cor}
\begin{proof}
  From Theorems \ref{thm:stationarity} and \ref{thm:stationarity_2}, we saw that $\widebar{X}$ is a KKT point of both (\ref{eq:prob_frac}) and its auxiliary problem (\ref{eq:aux_prob}) for $\rho=r(\widebar{X})$. Since the only KKT point of (\ref{eq:aux_prob}) is its unique minimum, $\widebar{X}$ solves (\ref{eq:aux_prob}) for $\rho=r(\widebar{X})$. From Theorem \ref{thm:global_min}, we conclude that $\widebar{X}$ also solves the global problem (\ref{eq:prob_frac}).
\end{proof}

\subsection{Convergence to Stationary Points and Global Minima}
Conditions \ref{cond:lin_indep_1} and \ref{cond:lin_indep_2} allowed us to show that fixed points of the F-DASF algorithm are stationary points of (\ref{eq:prob_frac}). It remains to show that the sequence $(X^i)_i$ converges to a global minimum of Problem (\ref{eq:prob_frac}), or at least to a stationary point. As the results provided in this subsection, and their proof, are similar to the ones of the original DASF algorithm, we will omit technical details and refer the reader to \cite{musluoglu2022unifiedp2}. To show convergence of $(X^i)_i$ to a single point, we require the following two conditions.

\begin{cond}\label{cond:continuity}
The local auxiliary problems (\ref{eq:loc_aux_prob_compact}) satisfy Assumptions \ref{asmp:well_posed} and \ref{asmp:kkt}.
\end{cond}

\begin{cond}\label{cond:finite_stat}
  The number of stationary points of each local auxiliary problem (\ref{eq:loc_aux_prob}) is finite, or the solver used by Algorithm \ref{alg:f_dasf} to solve the local auxiliary problems (\ref{eq:loc_aux_prob}) can only find a finite subset of the solutions of (\ref{eq:loc_aux_prob}).
\end{cond}

\noindent Condition \ref{cond:continuity} translates to requiring the local auxiliary problems to inherit the same properties of the centralized auxiliary problems (\ref{eq:aux_prob}). This condition is usually satisfied in practice since we already assumed that the centralized auxiliary problems (\ref{eq:aux_prob}) satisfy the assumptions presented in Section \ref{sec:assumptions}, while the local auxiliary problems are compressed versions of these. When Condition \ref{cond:continuity} is satisfied, any accumulation point $\widebar{X}$ of the sequence $(X^i)_i$ is a fixed point of Algorithm \ref{alg:f_dasf} for any $q$ and $\lim_{i\rightarrow+\infty}||X^{i+1}-X^i||=0$. In addition to Condition \ref{cond:continuity}, if Condition \ref{cond:finite_stat} is satisfied, then $(X^i)_i$ converges to a single fixed point $\widebar{X}$. For the proofs in \cite{musluoglu2022unifiedp2} to work in the case of F-DASF as well, we have to establish that all the points of the sequence $(X^i)_i$ lie in a compact set. This can be easily shown to be true for the F-DASF algorithm. Indeed, since the set $\mathcal{S}$ of Problem (\ref{eq:prob_frac}) is compact, and from the fact that all points of the sequence $(X^i)_i$ generated by the F-DASF algorithm lie in $\mathcal{S}$ (Lemma \ref{lem:X_in_constraints}), the sequence $(X^i)_i$ lies in a compact set. 

\begin{rem}
Note that in various fractional problems, the number of stationary points is not finite, for example in the TRO problem in Table \ref{tab:ex_prob}. It can be shown that a solution of the auxiliary problems of the TRO problem is given by an $X$ containing in its columns the $Q$ eigenvectors of $R_{\mathbf{vv}}-\rho^i\cdot R_{\mathbf{yy}}$ corresponding to its $Q$ largest eigenvalues. However, any matrix $XU$, where $U$ is an orthogonal matrix, is also a solution of the auxiliary problem of the TRO. For these cases, Condition \ref{cond:finite_stat} can be relaxed so as to require that the solver used for solving the auxiliary problems (\ref{eq:aux_prob}) can only obtain a finite set of the solutions of (\ref{eq:aux_prob}). For the TRO example, a solver that outputs the $Q$ principal unit norm eigenvectors of $R_{\mathbf{vv}}-\rho^i\cdot R_{\mathbf{yy}}$ is therefore sufficient, as the possible outputs are only different up to a sign change in each column, making the set of possible solutions finite (except for the contrived degenerate case where these $Q$ eigenvalues are not all distinct).
\end{rem}

The final step is to combine the previous results to be able to state the convergence guarantees of the F-DASF algorithm to a stationary point of (\ref{eq:prob_frac}). From Condition \ref{subcond:lin_indep} (either the form a or b), we know that fixed points of the F-DASF algorithm are stationary points of (\ref{eq:prob_frac}) satisfying its KKT conditions in Theorems \ref{thm:stationarity} and \ref{thm:stationarity_2}. Combining this with Conditions \ref{cond:continuity} and \ref{cond:finite_stat}, we can establish that the sequence $(X^i)_i$ converges and $\lim_{i\rightarrow+\infty}X^i=\widebar{X}$, where $\widebar{X}$ is a stationary point of (\ref{eq:prob_frac}) satisfying its KKT conditions.

Additionally, we can obtain stronger results if certain uniqueness assumptions are satisfied. Namely, if the global problem (\ref{eq:prob_frac}) has a unique minimum $X^*$ and no other stationary points, or if the conditions of Corollary \ref{cor:fixed_min} are satisfied, we have $\widebar{X}=X^*$, implying that $(X^i)_i$ converges to $X^*$. Finally, even in cases where these uniqueness assumptions are not met, we can still expect the F-DASF algorithm to converge to a global minimum if all minima are global minima. This is because of the monotonic decrease of $(r(X^i))_i=(\rho^i)_i$ (see Lemma \ref{thm:monotonic}), which implies that the sequence $(X^i)_{i}$ is kicked out of a potential equilibrium point which is not a minimum and cannot return to it, making the fixed points of Algorithm \ref{alg:f_dasf} that are in $\mathcal{X}^*$ the only stable ones.

In general, we see that the F-DASF algorithm converges under similar technical conditions as the DASF algorithm, except for the fact that some of these conditions are required to hold for the auxiliary problems (instead of the original problem) in the F-DASF case.

\section{Simulations}\label{sec:simulations}
In this section, we demonstrate the performance of the F-DASF algorithm in various experimental settings and compare it to the DASF algorithm for various fractional programs. Implementations are provided in \cite{musluoglu2022dsfotoolbox}. Throughout these experiments, we consider network topologies generated randomly using the Erd\H{o}s-R\'enyi model with connection probability $0.8$, and where the pruning function $\mathcal{T}^i(\cdot,q)$ is chosen to be the shortest path. Excluding Section \ref{sec:adaptive}, we consider stationary signals $\mathbf{y}$ and $\mathbf{v}$, following a mixture model:
\begin{align}
  \mathbf{y}(t)&=\Pi_s\cdot \mathbf{s}(t)+\mathbf{n}(t),\label{eq:signal_model_y}\\
  \mathbf{v}(t)&=\Pi_r\cdot \mathbf{r}(t)+\mathbf{y}(t)\nonumber \\
  &=\Pi_r\cdot \mathbf{r}(t)+\Pi_s\cdot \mathbf{s}(t)+\mathbf{n}(t),
\end{align}
with $\mathbf{r}(t)$, $\mathbf{s}(t)\overset{i.i.d.}{\sim}\mathcal{N}(0,\sigma_r^2)$, $\mathbf{n}(t)\overset{i.i.d.}{\sim}\mathcal{N}(0,\sigma_n^2)$ for every entry and time instance $t$. The entries of the mixture matrices $\Pi_s$ and $\Pi_r$ are independent of time and are independently drawn from $\mathcal{N}(0,\sigma_\Pi^2)$. For each experiment, the number of channels $M_k$ of the signals measured at node $k$ are equal for each node, and given by $M/K$, while each node measures $N=10^4$ samples of its local signals at each iteration of the algorithms. Table \ref{tab:table_param} gives an overview of the different parameters selected for the simulations presented below. The main performance metric we use to assess convergence of the F-DASF algorithm is the median squared error (MedSE) $\epsilon$:
\begin{equation}\label{eq:epsilon_error}
  \epsilon(i)=\text{median}\left(\frac{||X^i-X^*||_F^2}{||X^*||_F^2}\right),
\end{equation}
where the median is taken over multiple Monte Carlo runs and $X^*$ denotes an optimal solution of the respective problem, obtained from a centralized solver implementing Dinkelbach's procedure. If the centralized problem has multiple solutions, we select $X^*$ a posteriori as the one that best matches $X^i$ in the final iteration of the F-DASF algorithm. The results we present next have been obtained by taking the median over $100$ Monte Carlo runs, and a stopping criterion for Dinkelbach's procedure used in DASF has been set to be a threshold of $10^{-8}$ on the norm of the difference of two consecutive $\widetilde{X}_q$'s and a maximum number of iterations of $10$. 

We first consider two problems with compact constraint sets in a stationary setting in Sections \ref{sec:rtls} and \ref{sec:tro}. In Section \ref{sec:adaptive}, we will consider time-varying mixture matrices to simulate non-stationarity in an adaptive context, while Section \ref{sec:qol} demonstrates convergence for a problem with a non-compact constraint set.

\begin{table}[!t]
  \renewcommand{\arraystretch}{1.5}
  \caption{Summary of parameters used in the simulations.}
  \label{tab:table_param}
  \centering
  \begin{tabularx}{0.48\textwidth}{ >{\centering\arraybackslash}c | >{\centering\arraybackslash}X | >{\centering\arraybackslash}X | >{\centering\arraybackslash}X}
  \hline
  Experiment & Section \ref{sec:rtls}  & Section \ref{sec:tro}  &  Section \ref{sec:qol}  \\ \hhline{=|=|=|=}
  $Q$ &  $1$ &  $2$   &  $2$ \\ \hline
  $K$ &  \multicolumn{3}{c}{$10$} \\ \hline
  $M$ & $M=50$ & $M=50$ & $M=100$ \\ \hline
  Signal Statistics & $\sigma_r^2=0.5$, $\sigma_n^2=0.2$, $\sigma_\Pi^2=0.3$ & $\sigma_r^2=0.5$, $\sigma_n^2=0.1$, $\sigma_\Pi^2=0.1$   &   $\sigma_r^2=0.5$, $\sigma_n^2=0.2$, $\sigma_\Pi^2=0.2$ \\ \hline
  $N$ & \multicolumn{3}{c}{$10000$} \\ \hline
  Monte Carlo Runs & \multicolumn{3}{c}{$100$} \\ \hline
  \end{tabularx}
\end{table}

\subsection{Regularized Total Least Squares}\label{sec:rtls}
The regularized total least squares (RTLS) problem \cite{sima2004regularized,beck2006finding} is given as
\begin{equation}\label{eq:rtls}
  \begin{aligned}
  \underset{\mathbf{x}\in\mathbb{R}^M}{\text{min. } } & \frac{\mathbb{E}[|\mathbf{x}^T\mathbf{y}(t)-d(t)|^2]}{1+\mathbf{x}^T\mathbf{x}}=\frac{\mathbf{x}^TR_{\mathbf{yy}}\mathbf{x}-2\mathbf{x}^T\mathbf{r}_{\mathbf{y}d}+r_{dd}}{1+\mathbf{x}^T\mathbf{x}}\\
  \textrm{s. t. } & ||\mathbf{x}^TL||^2\leq 1,
  \end{aligned}
\end{equation}
where the variable $X=\mathbf{x}$ is a vector, i.e., $Q=1$. The matrix $R_{\mathbf{yy}}=E[\mathbf{y}(t)\mathbf{y}^T(t)]$ is the covariance matrix of $\mathbf{y}$, and similarly, we have $\mathbf{r}_{\mathbf{y}d}=\mathbb{E}[d(t)\mathbf{y}(t)]$ and $r_{dd}=\mathbb{E}[d^2(t)]$. In this example, (\ref{eq:signal_model_y}) is of the form $\mathbf{y}(t)=\mathbf{p}_s\cdot s(t)+\mathbf{n}(t)$, where $\Pi_s=\mathbf{p}_s$ is a vector and $s$ is a scalar. Moreover, $d$ represents a noisy version of $s$, given by $d(t)=s(t)+w(t)$, where each time sample of $w$ is drawn from $\mathcal{N}(0,0.02)$. Finally, $L$ is a diagonal matrix where each element of the diagonal follows $\mathcal{N}(1,0.1)$. For a fixed $\rho$, the auxiliary problem of (\ref{eq:rtls}) is
\begin{equation}\label{eq:rtls_aux}
  \begin{aligned}
    \underset{\mathbf{x}\in\mathbb{R}^{M}}{\text{minimize }} \quad &\mathbf{x}^TR_{\mathbf{yy}}\mathbf{x}-2\mathbf{x}^T\mathbf{r}_{\mathbf{y}d}+r_{dd}-\rho\cdot(1+\mathbf{x}^T\mathbf{x})\\
  \textrm{subject to} \quad & ||\mathbf{x}^TL||^2\leq 1.
  \end{aligned}
\end{equation}
The local auxiliary problem at iteration $i$ solved at the updating node $q$ in the F-DASF algorithm is then given by
\begin{equation}\label{eq:rtls_aux_loc}
  \resizebox{.48\textwidth}{!}{%
  $\begin{aligned}
  \underset{\widetilde{\mathbf{x}}_q\in\mathbb{R}^{\widetilde{M}_q}}{\text{min. } } & \widetilde{\mathbf{x}}_q^TR^i_{\widetilde{\mathbf{y}}_q\widetilde{\mathbf{y}}_q}\widetilde{\mathbf{x}}_q-2\widetilde{\mathbf{x}}_q^T\mathbf{r}^i_{\widetilde{\mathbf{y}}_qd}+r_{dd}-\rho^i(1+\widetilde{\mathbf{x}}_q^T\widetilde{I}_q^{i}\widetilde{I}_q^{iT}\widetilde{\mathbf{x}}_q)\\
  \textrm{s. t. } & ||\widetilde{\mathbf{x}}_q^T\widetilde{L}_q^i||^2\leq 1,
  \end{aligned}$%
  }
\end{equation}
where we have $R^i_{\widetilde{\mathbf{y}}_q\widetilde{\mathbf{y}}_q}=\mathbb{E}[\widetilde{\mathbf{y}}_q^i(t)\widetilde{\mathbf{y}}_q^{iT}(t)]$ and $\mathbf{r}^i_{\widetilde{\mathbf{y}}_qd}=\mathbb{E}[\widetilde{\mathbf{y}}_q^i(t)d(t)]$, with $\widetilde{\mathbf{y}}_q^i$ the locally available signal defined in (\ref{eq:tree_data}). The matrices $\widetilde{I}_q^i$ and $\widetilde{L}_q^i$ are the locally available versions of $I_M$ and $L$, respectively, obtained by taking $B=I_M$ and $B=L$ and computing $\widetilde{B}_q^i$ as in (\ref{eq:B_tilde}). Problem (\ref{eq:rtls_aux_loc}) is a quadratic problem with quadratic constraints and can be solved by a solver implementing, for example, an interior-point method. While F-DASF will only solve one instance of (\ref{eq:rtls_aux_loc}), the DASF algorithm will solve multiple problems of the form (\ref{eq:rtls_aux_loc}) at each iteration, following the steps of the Dinkelbach algorithm. Figure \ref{fig:rtls_prob_plot} shows the results of the comparison between F-DASF and DASF. We see that, by construction, the F-DASF algorithm requires solving fewer problems, yet the convergence rates are similar between F-DASF and DASF. More specifically, the F-DASF algorithm requires solving $5$ times fewer auxiliary problems (on average over iterations of median values) than the DASF algorithm.

Note that the sharp change in the convergence plot at iteration $i=10$ corresponds to the number of nodes $K=10$, and is due to the random initialization of $X^0$, as explained in Section \ref{sec:prob_params}.

\begin{figure}[t]
  \includegraphics[width=0.48\textwidth]{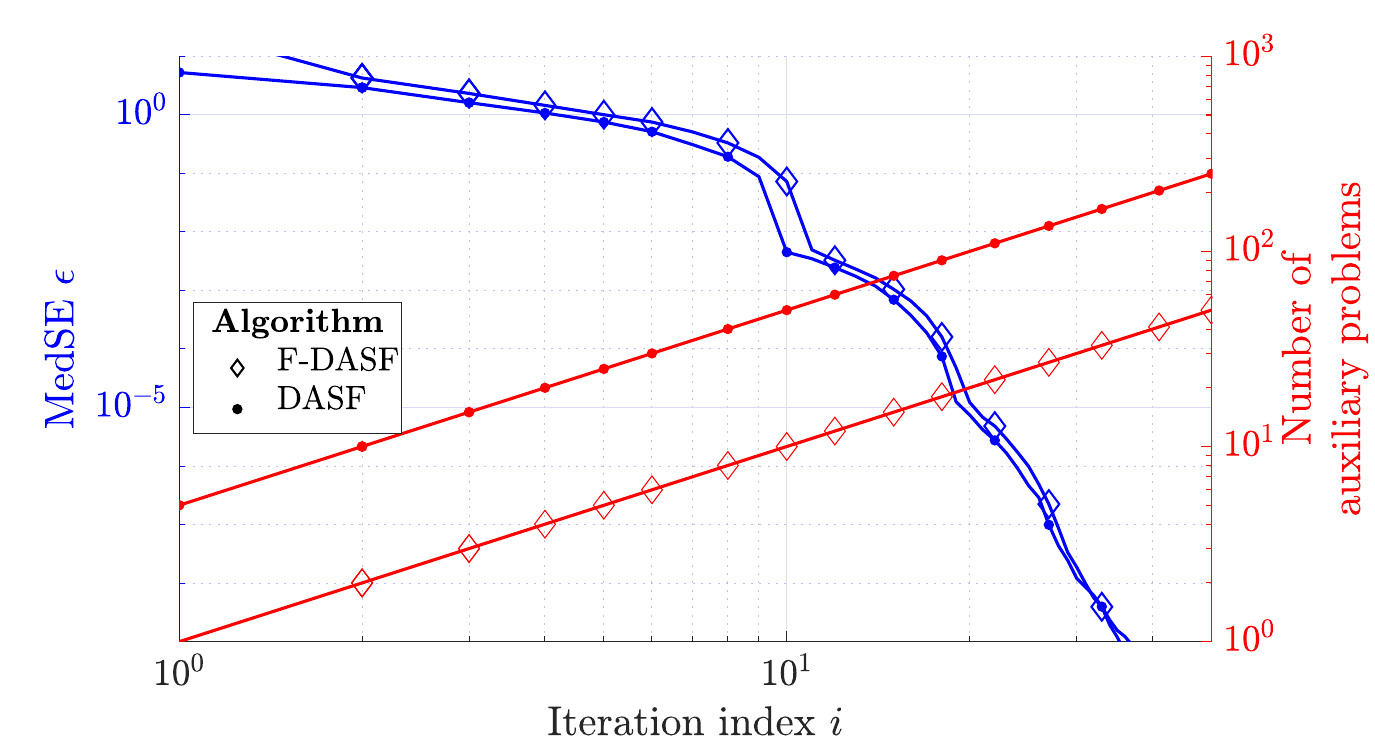}
  \caption{Convergence and cumulative computational cost comparison between the proposed F-DASF algorithm and the DASF algorithm when solving Problem (\ref{eq:rtls}).}
  \label{fig:rtls_prob_plot}
\end{figure}

\subsection{Trace Ratio Optimization}\label{sec:tro}
In this experiment, we consider the trace ratio optimization (TRO) problem and compare an F-DASF implementation with a DASF implementation\footnote{We note that the F-DASF implementation of TRO leads to the so-called DTRO algorithm in \cite{musluoglu2021distributed}, which can be viewed as a special case of the more general F-DASF algorithm.}. Here, we compare its performance to the DASF method. The TRO problem can be written as
\begin{equation}\label{eq:prob_tro}
  \begin{aligned}
    \underset{X\in\mathbb{R}^{M\times Q}}{\text{maximize }} \quad &\frac{\mathbb{E}[||X^T\mathbf{v}(t)||^2]}{\mathbb{E}[||X^T\mathbf{y}(t)||^2]}=\frac{\text{tr}(X^TR_{\mathbf{vv}}X)}{\text{tr}(X^TR_{\mathbf{yy}}X)}\\
  \textrm{subject to} \quad & X^TX=I_Q,
  \end{aligned}
\end{equation}
with $R_{\mathbf{vv}}=E[\mathbf{v}(t)\mathbf{v}^T(t)]$ and $R_{\mathbf{yy}}=E[\mathbf{y}(t)\mathbf{y}^T(t)]$, while for a given $\rho$, its auxiliary problem takes the form
\begin{equation}\label{eq:prob_tro_aux}
  \begin{aligned}
    \underset{X\in\mathbb{R}^{M\times Q}}{\text{maximize }} \quad &\text{tr}(X^TR_{\mathbf{vv}}X)-\rho\cdot\text{tr}(X^TR_{\mathbf{yy}}X)\\
  \textrm{subject to} \quad & X^TX=I_Q.
  \end{aligned}
\end{equation}
Considering $R_{\mathbf{yy}}$ and $R_{\mathbf{vv}}$ to be positive definite, a solution of (\ref{eq:prob_tro_aux}) is given by
\begin{equation}
  X_{\rho}=\text{EVD}_Q(R_{\mathbf{vv}}-\rho\cdot R_{\mathbf{yy}}),
\end{equation}
where $\text{EVD}_Q$ returns the $Q$ eigenvectors of $R_{\mathbf{vv}}-\rho\cdot R_{\mathbf{yy}}$ corresponding to its $Q$ largest eigenvalues. If the $Q+1$ eigenvalues are all distinct, problem (\ref{eq:prob_tro_aux}) satisfies the convergence assumptions. The local auxiliary problem solved using the F-DASF algorithm at each iteration $i$ and updating node $q$ is
\begin{equation}\label{eq:loc_prob_tro_aux}
  \begin{aligned}
    \underset{\widetilde{X}_q\in\mathbb{R}^{\widetilde{M}_q\times Q}}{\text{maximize }} \quad &\text{tr}(\widetilde{X}_q^TR^i_{\widetilde{\mathbf{v}}_q\widetilde{\mathbf{v}}_q}\widetilde{X}_q)-\rho^i\cdot\text{tr}(\widetilde{X}_q^TR^i_{\widetilde{\mathbf{y}}_q\widetilde{\mathbf{y}}_q}\widetilde{X}_q)\\
  \textrm{subject to} \quad & \widetilde{X}_q^TC_q^{iT}C_q^i\widetilde{X}_q=I_Q,
  \end{aligned}
\end{equation}
with $R^i_{\widetilde{\mathbf{v}}_q\widetilde{\mathbf{v}}_q}=E[\widetilde{\mathbf{v}}^i_q(t)\widetilde{\mathbf{v}}^{iT}_q(t)]$ and $R^i_{\widetilde{\mathbf{y}}_q\widetilde{\mathbf{y}}_q}=E[\widetilde{\mathbf{y}}^i_q(t)\widetilde{\mathbf{y}}^{iT}_q(t)]$, where $\widetilde{\mathbf{y}}_q^i$ and $\widetilde{\mathbf{v}}_q^i$ are defined as in (\ref{eq:tree_data}). A solution $\widetilde{X}_q^{*}$ of (\ref{eq:loc_prob_tro_aux}) is given by $\text{GEVD}_Q(R^i_{\widetilde{\mathbf{v}}_q\widetilde{\mathbf{v}}_q}-\rho^i\cdot R^i_{\widetilde{\mathbf{y}}_q\widetilde{\mathbf{y}}_q},C_q^{iT}C_q^i)$, where $\text{GEVD}_Q(C_1,C_2)$ returns the $Q$ generalized eigenvectors of the pair $(C_1,C_2)$ corresponding to its largest generalized eigenvalues. The DASF algorithm solves (\ref{eq:loc_prob_tro_aux}) multiple times at each updating node until the stopping criterion of the Dinkelbach procedure has been achieved. A comparison of the F-DASF and DASF algorithms is provided in Figure \ref{fig:tro_prob_plot}, where we again observe similar MedSE values per iteration for both methods, while the DASF algorithm requires solving $4.74$ times more auxiliary problems on average (over iterations of median values) than the F-DASF method.

\begin{figure}[t]
  \includegraphics[width=0.48\textwidth]{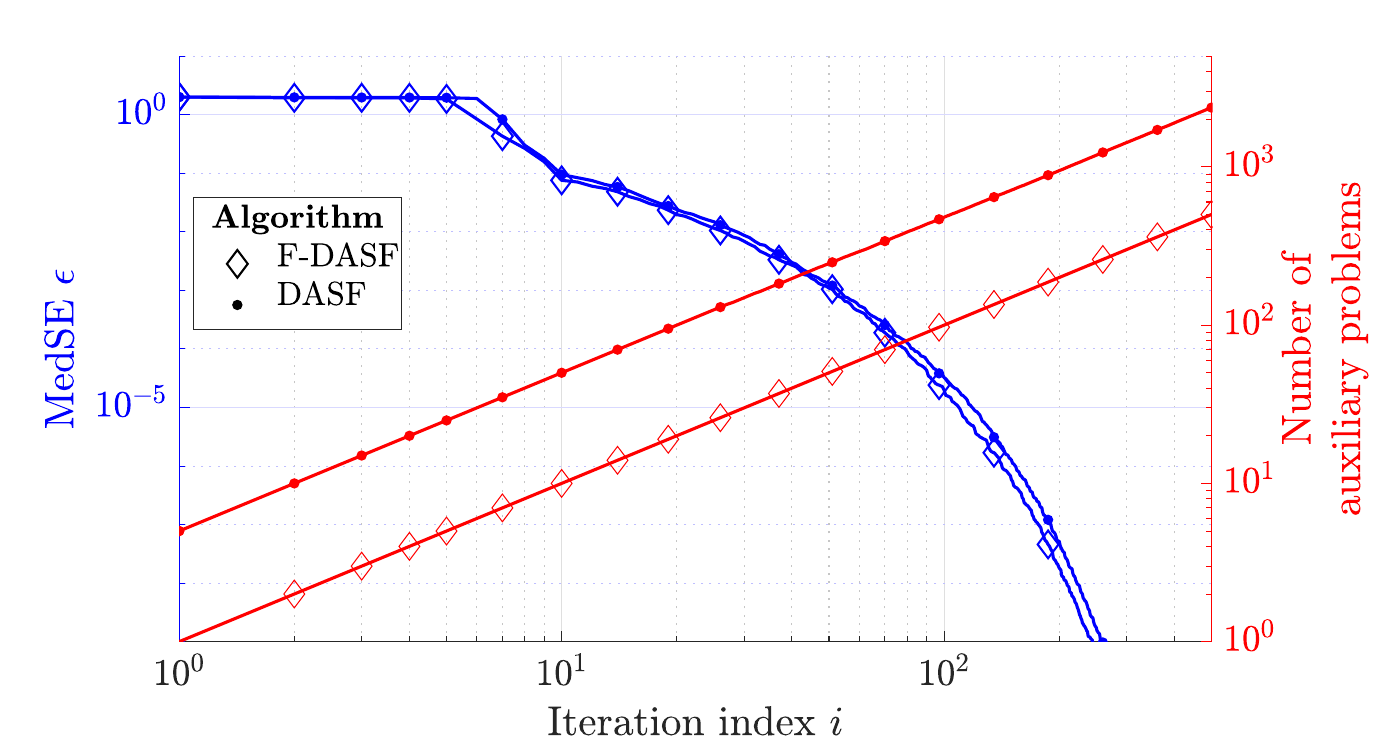}
  \caption{Convergence and cumulative computational cost comparison between the proposed F-DASF algorithm and the DASF algorithm when solving Problem (\ref{eq:prob_tro}).}
  \label{fig:tro_prob_plot}
\end{figure}

\subsection{F-DASF in an Adaptive Setting}\label{sec:adaptive}
We now consider the same settings as in Section \ref{sec:tro}, however, the signal models are now given by
\begin{align}
  \mathbf{y}(t)&=\Pi_s(t)\cdot \mathbf{s}(t)+\mathbf{n}(t),\label{eq:signal_model_y_ada}\\
  \mathbf{v}(t)&=\Pi_r(t)\cdot \mathbf{r}(t)+\Pi_s(t)\cdot \mathbf{s}(t)+\mathbf{n}(t),
\end{align}
i.e., the mixture matrices $\Pi_s$ and $\Pi_r$ are now time-dependent. In particular, we have $\Pi_s(t)=\Pi_{s,0}\cdot(1-p(t))+(\Pi_{s,0}+\Delta_{s})\cdot p(t)$, where $p$ is given in Figure \ref{fig:tro_prob_plot_adaptive} and the elements of $\Pi_{s,0}$ and $\Delta_s$ are independently drawn from $\mathcal{N}(0,0.1)$ and $\mathcal{N}(0,10^{-3})$ respectively, while $\Pi_r(t)$ is defined in the same way as $\Pi_s(t)$, except the fact that the elements of $\Pi_{r,0}$ are independently drawn from $\mathcal{N}(0,0.5)$. Therefore, the signals $\mathbf{y}$ and $\mathbf{v}$ are not stationary anymore, which implies that $X^*$ is time-dependent. In particular, at each iteration $i$, we estimate the solution $X^{*i}$ at iteration $i$ using $N=10^3$ time samples of $\mathbf{y}$ and $\mathbf{v}$ by solving
\begin{equation}\label{eq:prob_tro_adaptive}
  \begin{aligned}
    \underset{X\in\mathbb{R}^{M\times Q}}{\text{maximize }} \quad &\frac{\text{tr}(X^T\widehat{R}^i_{\mathbf{vv}}X)}{\text{tr}(X^T\widehat{R}^i_{\mathbf{yy}}X)}\\
  \textrm{subject to} \quad & X^TX=I_Q,
  \end{aligned}
\end{equation}
using the Dinkelbach procedure, where
\begin{align}
  \widehat{R}^i_{\mathbf{yy}}=\frac{1}{N}\sum_{\tau=0}^{N-1}\mathbf{y}(\tau+iN)\mathbf{y}^T(\tau+iN),
\end{align}
 and similarly for $\widehat{R}^i_{\mathbf{vv}}$. The MedSE $\epsilon$ is then given by
\begin{equation}
  \epsilon(i)=\text{median}\left(\frac{||X^i-X^{*i}||^2_F}{||X^{*i}||^2_F}\right).
\end{equation}
Figure \ref{fig:tro_prob_plot_adaptive} shows the value of $\epsilon$ over time $t$ with $i=\lfloor t/N\rfloor$, where we see that the F-DASF algorithm is able to track slow changes in the signal statistics of $\mathbf{y}$ and $\mathbf{v}$ and can adapt to abrupt changes as well, which is characterized by an initial jump in the MedSE value that gradually decreases. Note that $\epsilon$ settles around certain values due to the fact that $X^*$ is time-dependent, with a higher MedSE settling value for faster rates of change (i.e., steeper slope of $p$) in the signal statistics.

\begin{figure}[t]
  \includegraphics[width=0.48\textwidth]{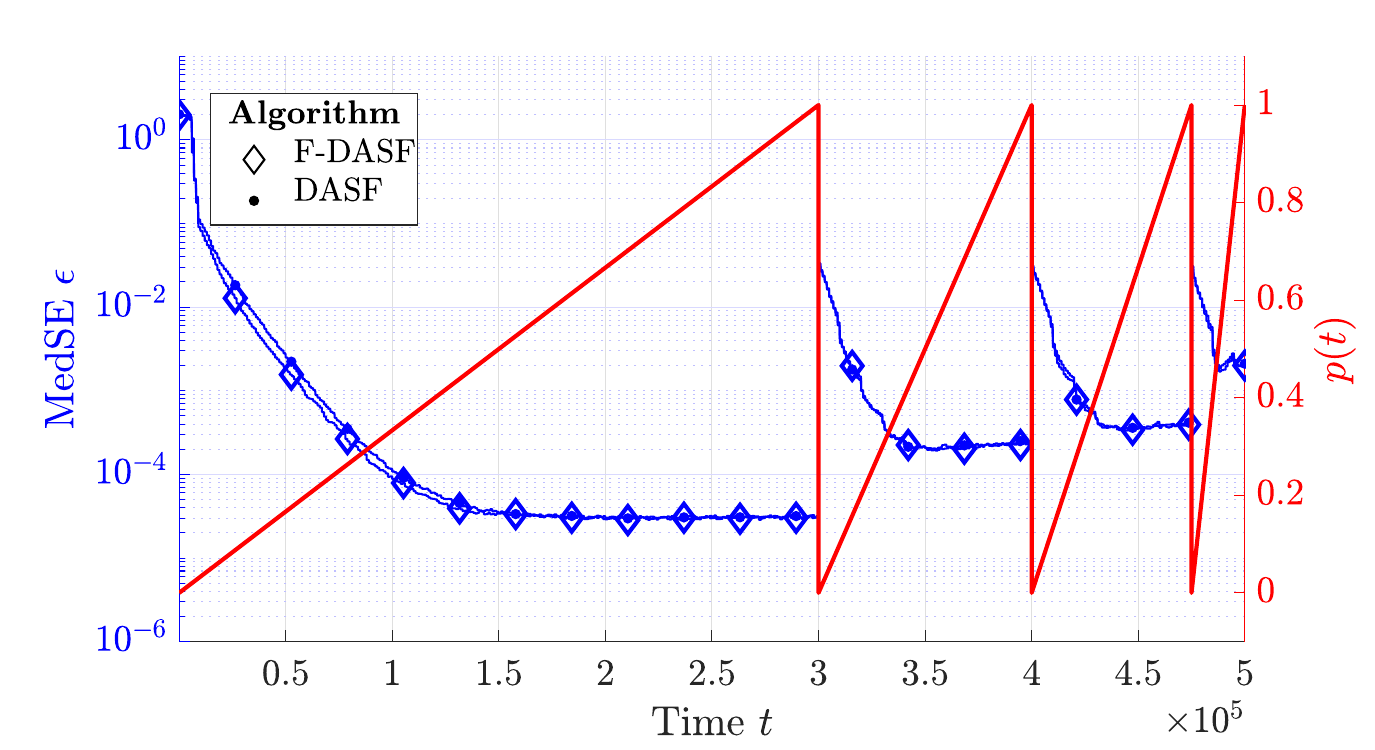}
  \caption{MedSE of the DASF and F-DASF algorithms when solving Problem (\ref{eq:prob_tro}) in an adaptive setting. The relationship between the time $t$ and the iterations $i$ is given by $i=\lfloor t/N\rfloor$.}
  \label{fig:tro_prob_plot_adaptive}
\end{figure}

\subsection{Quadratic over Linear}\label{sec:qol}
To demonstrate the generic nature of the (F-)DASF algorithm, let us consider the following arbitrary toy problem
\begin{equation}\label{eq:prob_qol}
  \begin{aligned}
    \underset{X\in\mathbb{R}^{M\times Q}}{\text{minimize }} \quad &\frac{\mathbb{E}[||X^T\mathbf{y}(t)||^2]+\text{tr}(X^TA)}{\text{tr}(X^TB)+c}\\
  \textrm{subject to} \quad & \text{tr}(X^TB)+c>0,
  \end{aligned}
\end{equation}
where $R_{\mathbf{yy}}=E[\mathbf{y}(t)\mathbf{y}^T(t)]$ is the covariance matrix of $\mathbf{y}$, assumed to be positive definite. In this example, every entry of the matrices $A\in\mathbb{R}^{M\times Q}$ and $B\in\mathbb{R}^{M\times Q}$ have been independently drawn from $\mathcal{N}(0,1)$, while $c$ is taken such that the problem is feasible\footnote{It can be shown that the problem is feasible, i.e., the solution is real, if $c$ satisfies either $2c\geq \text{tr}(A^TR_{\mathbf{yy}}^{-1}B)+\sqrt{\text{tr}(A^TR_{\mathbf{yy}}^{-1}A)\text{tr}(B^TR_{\mathbf{yy}}^{-1}B)}$ or $2c\leq \text{tr}(A^TR_{\mathbf{yy}}^{-1}B)-\sqrt{\text{tr}(A^TR_{\mathbf{yy}}^{-1}A)\text{tr}(B^TR_{\mathbf{yy}}^{-1}B)}$.}. It can be shown that the solution of Problem (\ref{eq:prob_qol}) has the form
\begin{equation}\label{eq:sol_qol}
  X^*=-\frac{1}{2}R_{\mathbf{yy}}^{-1}(A+\mu\cdot B),
\end{equation}
where $\mu$ is a scalar depending on $R_{\mathbf{yy}}$, $A$, $B$, and $c$. Note that the constraint $\text{tr}(X^TB)+c>0$ enforces the requirement $f_2(X)>0$, but does not constitute a compact set since it is open and unbounded, implying that neither the Dinkelbach procedure nor the F-DASF algorithm have a guarantee of convergence to an optimal point. Therefore, the convergence of the F-DASF (and DASF) algorithm will be assessed by comparing $X^i$'s to $X^*$ given in (\ref{eq:sol_qol}).

For a given $\rho$, the auxiliary problem of (\ref{eq:prob_qol}) is
\begin{equation}\label{eq:prob_qol_aux}
  \begin{aligned}
    \underset{X\in\mathbb{R}^{M\times Q}}{\text{minimize }} \quad &\text{tr}(X^TR_{\mathbf{yy}}X)+\text{tr}(X^TA)-\rho\cdot(\text{tr}(X^TB)+c)\\
  \textrm{subject to} \quad & \text{tr}(X^TB)+c>0.
  \end{aligned}
\end{equation}
If $R_{\mathbf{yy}}$ is positive definite, problem (\ref{eq:prob_qol_aux}) is convex and has a unique solution given by
\begin{equation}\label{eq:sol_qol_aux}
  X_{\rho}=\frac{1}{2}R_{\mathbf{yy}}^{-1}(\rho\cdot B-A).
\end{equation}

At each iteration $i$ of the F-DASF algorithm, the updating node $q$ solves its local auxiliary problem (\ref{eq:loc_aux_prob}) given by
\begin{equation}\label{eq:loc_prob_qol_aux}
\resizebox{.48\textwidth}{!}{%
    $\begin{aligned}
      \underset{\widetilde{X}_q\in\mathbb{R}^{\widetilde{M}_q\times Q}}{\text{minimize }} \quad &\text{tr}(\widetilde{X}_q^TR^i_{\widetilde{\mathbf{y}}_q\widetilde{\mathbf{y}}_q}\widetilde{X}_q)+\text{tr}(\widetilde{X}_q^T\widetilde{A}_q^i)-\rho^i\cdot(\text{tr}(\widetilde{X}_q^T\widetilde{B}_q^i)+c)\\
      \\
    \textrm{subject to} \quad & \text{tr}(\widetilde{X}_q^T\widetilde{B}_q^i)+c>0,
    \end{aligned}$%
    }
\end{equation}
where $R^i_{\widetilde{\mathbf{y}}_q\widetilde{\mathbf{y}}_q}=E[\widetilde{\mathbf{y}}^i_q(t)\widetilde{\mathbf{y}}^{iT}_q(t)]$ is the covariance matrix of the locally available stochastic signal $\widetilde{\mathbf{y}}_q^i$ at node $q$ defined in (\ref{eq:tree_data}), while $\widetilde{A}_q^i$ and $\widetilde{B}_q^i$ are the locally available deterministic matrices, as defined in (\ref{eq:B_tilde}). The solution $\widetilde{X}_q^{*}$ of (\ref{eq:loc_prob_qol_aux}) is obtained by replacing $(\rho, R_{\mathbf{yy}}, A, B)$ by $(\rho^i, R^i_{\widetilde{\mathbf{y}}_q\widetilde{\mathbf{y}}_q}, \widetilde{A}_q^i, \widetilde{B}_q^i)$ in (\ref{eq:sol_qol_aux}), since both (\ref{eq:prob_qol_aux}) and (\ref{eq:loc_prob_qol_aux}) have the same form. In contrast, the DASF algorithm will solve multiple problems (\ref{eq:loc_prob_qol_aux}) at each iteration $i$ and node $q$ as it will apply Dinkelbach's procedure on a compressed version of (\ref{eq:prob_qol}). Figure \ref{fig:qol_prob_plot} shows a comparison of both the MedSE values and the cumulative computational cost between the proposed F-DASF algorithm and the existing DASF method. Despite the fact that Problem (\ref{eq:prob_qol}) does not have a compact set, we see that convergence can still be achieved to the optimal solution $X^*$ given in (\ref{eq:sol_qol}), showing that the F-DASF algorithm can still be used to solve problems with non-compact constraint sets. The F-DASF algorithm solves an average value (over iterations of median values) of $5.77$ times fewer auxiliary problems than DASF for the case of Problem (\ref{eq:prob_qol}), while again obtaining similar convergence results.

\begin{figure}[t]
  \includegraphics[width=0.48\textwidth]{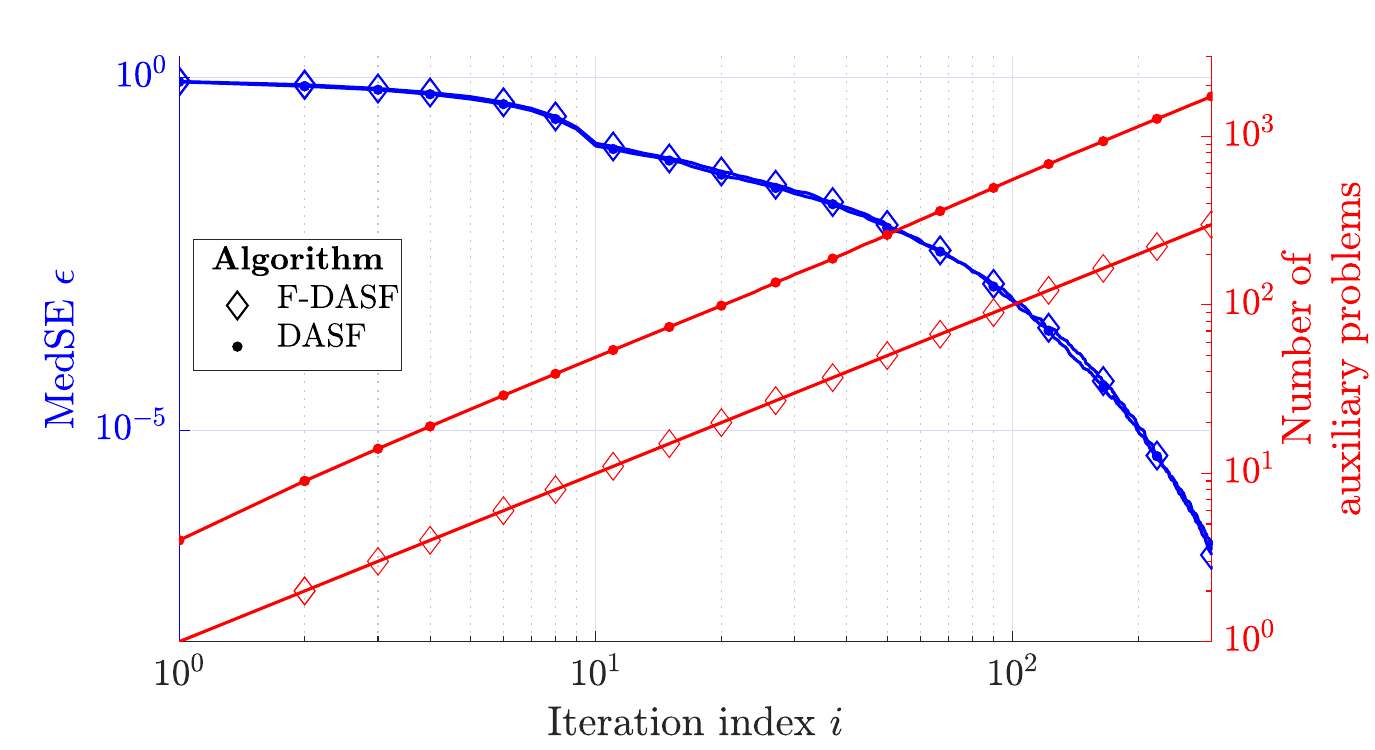}
  \caption{Convergence and cumulative computational cost comparison between the proposed F-DASF algorithm and the DASF algorithm when solving Problem (\ref{eq:prob_qol}).}
  \label{fig:qol_prob_plot}
\end{figure}

\section{Conclusion}
We have proposed the F-DASF algorithm which exploits the structure and properties of a fractional programming method to significantly reduce the computational cost of the DASF method applied to problems with fractional objectives. In particular, the DASF method requires solving a fractional program at each iteration while the F-DASF method only requires a partial solution, implying significantly fewer computations. Despite this reduced number of computations, the proposed method is shown to converge under similar assumptions as the DASF algorithm, and at similar convergence rates. The results obtained in this paper also show that partial solutions can be sufficient for the DASF method to converge and open the way for further studies on other families of problems fitting the DASF framework where computational cost reductions can be obtained by solving optimization problems only partially at each iteration.




\bibliographystyle{IEEEtran}
\bibliography{IEEEabrv,IEEEexample}

\end{document}